\documentclass[twocolumn]{article}
\usepackage{usenix2019_v3}
\usepackage{tikz}
\usepackage{amsmath}
\usepackage{filecontents}
\usepackage{fancyvrb}
\usepackage{xcolor}
\usepackage{amsfonts}
\usepackage{mathtools}
\usepackage{tikz}
\usepackage{amsmath}
\usepackage[utf8]{inputenc}
\usepackage[english]{babel}
\usepackage{amsthm}
\usepackage{algorithm}
\usepackage[noend]{algpseudocode}
\usepackage[titletoc]{appendix}
\usepackage{ifthen}
\usepackage{titlesec}
\definecolor{ao}{rgb}{0.0, 0.5, 0.0}
\theoremstyle{definition}
\theoremstyle{definition}\newtheorem{theorem}{Theorem}
\theoremstyle{definition}\newtheorem{lemma}{Lemma}
\theoremstyle{definition}\newtheorem{definition}{Definition}
\newtheorem{assumption}{Assumption}

\newtheorem*{problem}{MOWSP Problem}
\titleformat{\paragraph}
{\normalfont\normalsize\bfseries}{\theparagraph}{1em}{}
\usepackage{amsmath}

\DeclareMathOperator*{\argminA}{arg\,min} 

\begin{document}

\date{}

\title{\Large \bf An Efficient Algorithm for Finding Sets of Optimal Routes}

\author{
{\rm Ido Zoref and Ariel Orda}\\
Viterbi Faculty of Electrical Engineering , The Technion – Israel Institute of Technology
} 

\maketitle

\begin{abstract}
In several important routing contexts it is required to identify a set of routes, each of which optimizes a different criterion. For instance, in the context of vehicle routing, one route would minimize the total distance traveled, while other routes would also consider the total travel time or the total incurred cost, or combinations thereof. In general, providing such a set of diverse routes is obtained by finding optimal routes with respect to different sets of weights on the network edges. This can be simply achieved by consecutively executing a standard shortest path algorithm. However, in the case of a large number of weight sets, this may require an excessively large number of executions of such an algorithm, thus incurring a prohibitively large running time.
\\

We indicate that, quite often, the different edge weights reflect different combinations of some "raw" performance metrics (e.g., delay, cost). In such cases, there is an inherent dependency among the different weights of the same edge. This may well result in some similarity among the shortest routes, each of which being optimal with respect to a specific set of weights. In this study, we aim to exploit such similarity in order to improve the performance of the solution scheme. 
\\

Specifically, we contemplate edge weights that are obtained through different linear combinations of some (``raw'') edge performance metrics. We establish and validate a novel algorithm that efficiently computes a shortest path for each set of edge weights. We demonstrate that, under reasonable assumptions, the algorithm significantly outperforms the standard approach. Similarly to the standard approach, the algorithm iteratively searches for routes, one per set of edge weights; however, instead of executing each iteration independently, it reduces the average running time by skillfully sharing information among the iterations.
\\

\end{abstract}

\section{Introduction}\label{intro}
Routing is a core networking task, in particular in road networks and computer networks \cite{routingTask1} \cite{routingTask2} \cite{routingTask3} \cite{routingTask4}. A fundamental routing problem is that of finding a route that optimizes a single additive metric (e.g., time, cost). Essentially, its solution consists of computing a shortest path on a weighted graph.
\\

In some routing contexts, the routing objective function cannot be satisfactorily captured through a single additive edge cost. For example, planing a vehicle trajectory to a desired destination often requires a fast enough route (in terms of driving time) as well as a route that incurs sufficiently low fuel cost. Similarly, communication network routing problems often consider both throughput and delay objectives.
\\

The problem of routing with multiple objectives is a well-known challenge. Usually, different routes optimize different objective functions, hence finding a single route that simultaneously (i.e., independently) minimizes more than a single objective function (e.g. both travel time and cost) is usually impossible. We proceed to discuss the different alternatives for handling the multi-objective routing problem and present our selected approach.
\\

A common approach for dealing with multi-objective routing problems is to search for a single path that optimizes one selected objective, subject to the condition that the other objectives do not exceed given threshold values.
This is an NP-hard problem, known as the Resource Constrained Shortest Path (CSP) problem \cite{cspDef}.
In the literature, approximate algorithms and heuristic approaches have been proposed to deal with the CSP problem.\\
Approximation algorithms for the CSP problem are usually based on scaling and rounding of the input data. Warburton \cite{Warburton} was the first to develop a fully polynomial time approximation algorithm for the CSP problem on acyclic networks. In \cite{Hassin}, Hassin later improved upon this to derive two fully polynomial time approximation schemes (FPTAS) that are applicable for general networks. Other related approximation schemes providing certain improvements to Hassin’s algorithm can be found in \cite{LorenzoOrda}. In particular, a significant improvement of Hassin's result was achieved by Lorenz and Raz \cite{LorenzoRaz}, who established a strongly polynomial time approximation scheme that is also applicable to general networks.  In \cite{Goel}, the authors considered the problem of finding a path whose delay is at most (1 + $\epsilon$) times the specified delay bound and whose cost is not greater than that of the minimum cost path of the CSP problem.\\

As for heuristic approaches for the CSP problem, several proposed algorithms are based on solutions to the integer relaxation or the dual of the integer relaxation of the CSP problem \cite{IntegerCSP}. Juttner et al. \cite{Juttner} introduced the LARAC algorithm, which solves the Lagrangian relaxation of the CSP problem. An applicable heuristic variation of the LARAC algorithm is LARAC-BIN \cite{larac}. It employs binary search technique to skillfully find a path $p$ where the deviation of $p$'s cost from the optimal path cost is smaller then the tuning parameter $\tau$.
\\

An alternative approach for dealing with multi-objective routing is to provide a (``suitable'') {\em set of routes}. This can be obtained by finding several optimal routes with respect to different sets of edge weights. For instance, in the context of vehicle routing, one route may minimize the total distance traveled, while other routes would also consider the total travel time or the total incurred cost, or combinations thereof.
\\

Such an approach is often used  in recommendation systems, such as autonomous navigation systems \cite{EvRoutePlaning}, \cite{HAR}. In these systems, it is difficult to match the personal preferences of individual users by providing only a single route. These preferences may be based, for example, on better local knowledge, a bias for or against a certain route objective, or other factors. One can deal with this issue by presenting the user with a number of alternative routes with the hope that one of them would be satisfactory. 
\\

A possible method for identifying such a ("suitable") set of routes is to search for several routes that optimize different linear combinations of some "raw" objectives \cite{flexObj}, e.g., distance traveled, incurred cost, etc. Such a linear combination allows to translate several objectives into a single weight for each edge, which in turn implies an ``optimal'' (i.e., minimum weight) route from a source to target node. Therefore, by considering several linear combinations, finding the set of optimal routes (each defined with respect to one linear combination) can be simply achieved by consecutively executing a standard shortest path algorithm. However, in the case of a large number of linear combinations, this may require an excessively large number of executions of such an algorithm, thus incurring a prohibitively large running time.\\

Accordingly, in this study we seek to efficiently find a set of routes, each of which minimizes a different linear combination of the objectives. Funke and Storandt \cite{Funke} suggest such a solution. Specifically, they show that contraction hierarchies – a very powerful speed-up technique originally developed for the standard shortest path problem \cite{ch}, can be constructed efficiently for the case of several linear combinations of the edge "raw" objectives. However, this method requires some pre-processing efforts. Pre-processing can be applied to speed up graph-based search algorithms in car navigation systems or web-based route planners. Still, in some routing contexts it may not be possible to perform pre-processing since the network may dynamically change over short periods of time.
\\

A possible alternative solution for the considered problem can be obtained by finding a set of \textit{Pareto optimal} routes.\footnote{To define Pareto optimality, consider two $W$–dimensional objective vectors $x = (x_1,...,x_W)$ and $y = (y_1,...,y_W)$. If $x_i \leq y_i$ for each $j \in \{ 1,..,W \}$ and $x_i < y_i$ for some $j \in \{ 1,..,W \}$, then $x$ dominates $y$. Given a finite set $X$ of $W$-dimensional vectors, we say that $x \in X$ is {\em Pareto–optimal in $X$} if there is no $y \in X$ that dominates $x$.}  
\\
Recall that, in our problem, each required route minimizes a different linear combination of the "raw" objectives. It is easy to verify that each such route is actually a \textit{Pareto-Optimal} route. Hence, by finding the set of Pareto optimal routes we obtain a solution for the problem. We shall present several algorithms for finding the set of Pareto optimal routes.
\\

Standard solutions for finding the set of Pareto optimal routes are inspired by the Dijkstra algorithm (which is the standard approach for finding a shortest route in the basic single-objective case). The most common approaches are the Multicriteria Label-Setting (MLS) algorithm \cite{Martins} \cite{Hansen} and the Multi-Label-Correcting (MLC) algorithm \cite{Dean} \cite{DellingWagner}. The MLS algorithm keeps, for each node, a set of non-dominated paths. The priority queue maintains paths instead of vertices, typically ordered lexicographically. At each iteration, it extracts the minimum path, say $L$, which is a path ending at some node, say $u \in V$. Then, MLS scans the outgoing edges from $u$, say $a = (u, v)$ . It does so by evaluating whether the new path, $L_{new}=L \mathbin\Vert a$, is non-dominated by the queue's paths, and in that case it inserts the new path, $L_{new}$, to the queue. The MLC algorithm is quite similar, except that it sorts the priority queue in a more intuitive manner, thus resulting in better performance in term of running time. 
\\

Both MLS and MLC are considered to be fast enough as long as the set of Pareto paths is relatively small \cite{DisserParetoFeasible} \cite{paretoIsFeasible}. Unfortunately, Pareto sets may consist of a prohibitively large (non-polynomial) number of paths, even for the restricted case of two objectives \cite{Martins}. A possible approach for dealing with larger Pareto sets can be obtained by $\epsilon$-optimal solutions, which provide a polynomial number of Pareto paths \cite{Papadimitriou}. Such an approximated solution can be computed efficiently through a fully polynomial approximation scheme (FPAS) \cite{Loridan} \cite{Tsaggouris} \cite{White}.
However,  as explained above, we seek to find a sub-set of Pareto routes where each route minimizes a different linear combination of the "raw" objectives. Hence, obtaining a sub-optimal Pareto set may not provide a valid solution. In other words, it is not guaranteed that the set of required routes is included in an approximated Pareto set.
\\

Accordingly, in this study we establish and validate the \textit{Iterative Dijkstra with Adapted Queue (IDAQ)} algorithm. IDAQ provides an optimal solution for the problem without imposing any pre-processing computations. We show that, under reasonable assumptions, IDAQ significantly outperforms the standard approach of consecutively executing a standard shortest path algorithm. Similarly to the standard approach, IDAQ iteratively searches for routes; however, instead of executing each iteration independently, it shares information among the iterations. Unlike the standard approach, this allows IDAQ avoid scanning any optimal path (with respect to some linear combination) more than once, thus providing improved performance in terms of running time.
\\

The rest of the paper is organized as follows. In Section 2, we formulate the problem.
Next, in Section 3, we present and analyze a standard approach algorithm. Our approach, namely the IDAQ algorithm, along with a theoretical analysis, are presented in Sections 4 and 5, respectively. Section 6 presents a simulation study, which demonstrates that the IDAQ algorithm considerably improves performance (in terms of running time) in comparison to the standard approach. This is demonstrated on both randomly generated settings as well as on settings that correspond to real-life data. Finally, concluding remarks are presented in Section 7.

\section{Problem Formulation}
In this section, we formulate the \textit{Multi-Objective Weighted Shortest Path problem (MOWSP)}, discussed in this article.\\
To that end, we shall use the following definitions.
\begin{definition}[edge]\label{edge} 
Consider a set of nodes $V$. An $edge$ is defined as an ordered pair of nodes in $V$.
\\

As mentioned, we aim to solve a problem of finding several optimal routes with respect to different sets of edges weights. More specifically, each set of edges weights is produced through a (different) linear combination of the  "raw" objectives. Formally, $MOG$ (definition \ref{MOG}) is defined as a graph with several objective values attached to each edge; and a $coefficient$ $vector$ (definition \ref{lamda}) specifies the considered linear combinations of the objectives (where each coefficient represents the relative importance of each objective).
\\

\end{definition}
\begin{definition}[MOG]\label{MOG} A Multi-Objective Graph, $MOG(V,E)$, is a set of connected nodes $V = \{v_1,...,v_n\}$ and a set of directed edges $E = \{e_1,...,e_m\}$ so that associated with each edge $e \in E$ are $W$ different non-negative additive values, each representing some objective. We shall
denote the $i$th objective of an edge $e \in E$ (where $i \leq W$) by $w_e[i]$.
\end{definition}

\begin{definition}[coefficient vector]\label{lamda} A $coefficient$ $vector$ $\lambda \in  \mathbb{R} ^{W}$  is a vector of $W$ positive numbers. We denote the $i$th element of $\lambda$ (where $i \leq W$) by $\lambda [i]$.
\end{definition}

\begin{definition}[$Path_{s}$]\label{paths} Consider a $MOG(V,E)$ (definition \ref{MOG}) and a ``start node'' $s \in V$. $Path_{s}$ is the set of paths from $s$ to any node $v \in V $. In other words, each element of $Path_{s}$ is an ordered list of nodes that starts with $s$ and each subsequent node is connected to the previous one by an edge $e \in E$ \footnote{ Notice that under the $Path_{s}$, $Path_{s,v}$ definitions, we also consider paths that contain loops.}

\end{definition}
\begin{definition}[$Path_{s,v}$]\label{stateDef} Consider an $MOG(V,E)$ (definition \ref{MOG}), a start node $s \in V$ and some node $v \in V$. $Path_{s,v}$ is the set of paths from $s$ to $v$. In other words, each element $p \in Path_{s,v}$ is the prefix of $Path_{s}$ having $v$ as the last node.
\end{definition}

\begin{definition}[edge cost]

Consider a coefficient vector $\lambda$ (definition \ref{lamda}). The $cost$ of an edge $e \in E$ with respect to $\lambda$ is defined as follows:
\begin{equation}
\label{edgeEq}
    Cost(e, \lambda) =\sum_{j = 0 } ^{W}\lambda[j] \cdot w_e[j]
\end{equation} 
\end{definition}

\begin{definition}[path cost]
 Consider a coefficient vector $\lambda$ (definition \ref{lamda}), a start node $s \in V$ and a path $p \in Path_{s}$. The $cost$ of $p$ with respect to $\lambda$ is defined as follows:
\begin{equation}
\label{costEq}
    Cost(p, \lambda) =\sum_{e \in p} Cost(e,\lambda) = \\
    \sum_{e \in p}\sum_{j = 0 } ^{W}\lambda[j] \cdot w_e[j]
\end{equation} 
\end{definition}

\begin{definition}\label{opDef}[Optimal $Path_{s,v}$ with respect to $\lambda$] 
Consider an $MOG(V,E)$, a start node $s \in V$, a node $v \in V$ and a coefficient vector $\lambda$. $p_{opt} \in Path_{s,v}$ is $optimal$ with respect to $\lambda$ if, for each $p \in Path_{s,v}$,  the following holds:
\begin{equation}
\label{optimalEq}
    Cost(p_{opt}, \lambda)  \leq Cost(p, \lambda)
\end{equation} 
\end{definition}
We are ready to state the MOWSP problem:
\begin{problem}[Multi-Objective Weighted Shortest Path problem]\label{mowsp}
Consider an $MOG$ $G(V,E)$, a start node $s \in V$ and a set of coefficient vectors $\Lambda$, where:
\begin{equation}
\label{lambdaEq}
    \Lambda = \{ \lambda_{i} \in  \mathbb{R} ^{w} | i\in{1,..,K}\}
\end{equation}
For each coefficient vector $\lambda_i \in \Lambda$, find set of paths 
$$P_i=\{p_{v_1}, p_{v_2},...,p_{v_n}\}$$
where each path $p_{v_i} \in Path_{s,v_i}$ is optimal with respect to $\lambda_i$.
\end{problem}


\section{Standard Algorithm: Dijkstra Iterations}
\label{naiveAlgo}
In this section, we present a standard-approach algorithm for solving the MOWSP problem, which is based on multiple executions of the Dijkstra algorithm (hencefort, {\em the Standard Algorithm}). The rest of this section is structured as follows: in Section \ref{standartAlgoSec}, we describe the Standard Algorithm; in Section \ref{standardCorrect}, we prove its correctness; finally, in Section \ref{standardComplexity}, we analyze its time complexity.
\\
\subsection{Standard Algorithm Description}\label{standartAlgoSec}
The algorithm iterates through each $ \lambda_{i} \in \Lambda $ (defined by equation \ref{lambdaEq}). In each iteration, the algorithm constructs a new graph $G_{temp}(V, E)$, which is identical to $G(V, E)$ except for the following: we reduce the cost of each edge $e \in E$ to a single dimension using the Cost(e, $\lambda_{i} $) function (defined by equation \ref{edgeEq}). The algorithm calculates and returns the shortest path for each node in $G_{temp}(V, E)$ from the start node $s$ using the Dijkstra shortest-path algorithm.
\\
The pseudo-code of the Standard Algorithm is presented herein.
\begin{algorithm}
\caption{Standard Algorithm}\label{naiveAlgoPseudo}
\begin{algorithmic}[1]
\Procedure{Standard Algorithm}{G,s,$\lambda$}
\State $optimal\_paths$ $ \gets \text{null}$
\For{i=1 to K}
\State $E_{temp} \gets G.E$
\For{ each $e \in E_{temp}$}
\State $e.Cost \gets Cost(e, \lambda_{i})$
\EndFor
\State $G_{temp} \gets (G.V,E_{temp})$
\State $P_i \gets \text{Dijkstra(G$_{temp}$, s)}$
\State $optimal\_paths\_set \gets optimal\_paths\_set \cup \{P_i\}$
\EndFor

\State Return $optimal\_paths\_set$
\EndProcedure
\end{algorithmic}
\end{algorithm}

\subsection{Standard Algorithm Correctness}
\label{standardCorrect}
\begin{lemma}
The Standard Algorithm solves the MOWSP problem.
\end{lemma}
\begin{proof}

As explained in Section \ref{naiveAlgo}, The Standard Algorithm is based on $K$ executions of the Dijkstra algorithm. Consider one of these executions (lines 4-9 of some iteration).
\\

In lines 4-6, the graph $G_{temp}$ is generated, and on this graph, we execute the Dijkstra algorithm.
The nodes of $G_{temp}$ are precisely those of $G$, namely $V$. The edges of $G_{temp}$, namely $E_{temp}$, are equal to $E$ with the following addition: each edge $e \in E_{temp}$ cost value, denoted by $e.Cost$, is determined by the $Cost$ function (line 6).
\\

The Dijkstra algorithm (executed in line 8) returns for each node $v \in V$ a path, namely $p_v \in Path_{s,v}$. $p_v$ is the shortest path ending at $v$, with respect to the cost values $e.Cost$ for each edge $e \in E_{temp}$. 
\\
Let us denote the output of the Dijkstra algorithm as $P_i$ where: 
$$P_i=\{p_{v_1}, p_{v_2},...,p_{v_n}\}$$

Recall that $G_{temp}$ has the same nodes and edges as $G$ hence, for any node $v$, each path $p \in Path_{s,v}$ is necessarily a path on both $G_{temp}$ and $G$.
\\
From the correctness of the Dijkstra algorithm we can conclude that the path $p_{v}$ is such that:
$$
 p_{v}=\argminA_{p \in Path_{s,v}} f(p)
$$
where:
$$
    f(p) =\sum_{e \in p} e.Cost =\sum_{e \in p} Cost(e,\lambda_t) = Cost(p,\lambda_{t})
$$
In other words, for each $p \in Path_{s,v}$ the following equation holds:    
$$
Cost(p_{v}, \lambda_t)  \leq Cost(p, \lambda_t)
$$
Notice that, by definition, $p_{v}$ is  optimal with respect to $\lambda_i$, i.e,  $P_i$ is the set of optimal $Path_{s,v}$ for each $v \in V$ with respect to $\lambda_i$. 
\\
The $optimal\_paths\_set$ which is return at the end of the algorithm is precisely:
$$
optimal\_paths\_set=\{P_1,P_2,...,P_K\}
$$
hence, it is by definition the solution for the MOWSP as required.
\end{proof}

\subsection{Standard Algorithm Complexity Analysis}\label{standardComplexity}
\begin{lemma}
\label{l0}
The complexity of the Standard Algorithm is given by:
    $$O(K \cdot W \cdot |E|+K \cdot |V| \cdot log|V|)$$
\end{lemma}
\begin{proof}
Recall that $K$ is the number of input coefficient vectors (see MOWSP definition), and $W$ is defined as the number of objective values in each edge (see definition \ref{MOG}).
\\

As described in Section \ref{naiveAlgo}, The Standard Algorithm includes $K$ iterations, in each of which the following operations are performed:
\begin{enumerate}
  \item  Constructing a search graph G$_{temp}$ by calculating for each edge $e \in E$ its cost value using the $Cost$ function (equation \ref{costEq}). This is done in $ O(|E| \cdot W)$.
  \item Calculating the shortest path for each node in G$_{temp}$ using the Dijkstra algorithm. This can be done in $ O(|E|+|V| \cdot \log{|V|})$ using a Fibonacci heap for the priority queue implementation \cite{fibQueueDijkstraComplexity}.
\end{enumerate}
Hence, the complexity of the Standard Algorithm is given by:
    $$O(K \cdot W \cdot |E|+K \cdot |V| \cdot log|V|)$$
\end{proof}
\section{IDAQ Algorithm}
We turn to present IDAQ, an efficient algorithm that solves the MOWSP problem. We shall prove that, under reasonable assumptions, IDAQ has lower time complexity than the Standard Algorithm  described in Section \ref{naiveAlgo}.\\
The rest of this section is organized as follows. We begin with a general description of the IDAQ algorithm (Subsection \ref{IDAQDesc}). Then, in Subsections \ref{domSum} and \ref{adaptedQueue}, we introduce procedures that serve as building blocks by the IDAQ algorithm. Finally, in Subsection \ref{idaqPseudoSec} we present the IDAQ algorithm.
\subsection{General Description}
\label{IDAQDesc}
Similarly to the Standard Algorithm, IDAQ is an iterative algorithm: at the end of each iteration, say $i \in \{1,..,K\}$, IDAQ finds the the set of optimal $Path_{s,v}$ for each $v \in V$ with respect to $\lambda_i$. 
\\
However, unlike the Standard Algorithm, IDAQ shares knowledge among its iterations. The basic idea is the following: consider some node $v \in V$. While evaluating a path $p \in Path_{s,v}$, IDAQ checks whether $p$ might be optimal with respect to $\lambda_I$ where $I \in \{i+1,..,K\}$ is some future iteration. In that case, IDAQ remembers $p$, and this allows IDAQ to execute iteration $I$ more efficiently.\\

\subsection{IDAQ Relevance Check}
\label{domSum}
In this subsection we present IDAQ Relevance Check, a procedure used by IDAQ to determine whether an evaluated path is potentially optimal with respect to any coefficient vector $\lambda \in \Lambda$. Specifically, in Subsection \ref{relDefSec} we present the definition of {\em a relevant path}
and in Subsection \ref{isRelAlgoSec} we present the procedure's pseudo-code.

\subsubsection{Relevance definition}\label{relDefSec}
In this subsection we establish the relevance definition (definition \ref{nonDomDef}). Intuitively, a relevant path is a potential optimal path with respect to any coefficient vector $\lambda \in \Lambda$ and therefore should not be ignored. We begin by introducing some auxiliary definitions.

\begin{definition}[$Q_{v}$]\label{Qv} Consider some node, say $v \in V$, and  a list of paths $Q \subseteq Path_{s}$. $Q_{v}$ is defined as the sub-group of paths in $Q$ containing each paths $b \in Q$ where $b \in Path_{s,v}$.
\end{definition}

\begin{definition}[relevant due to optimality with respect to $Q_v$] \label{relOpt} Consider some node $v \in V$, a list of paths $Q_v \subseteq Path_{s,v}$ and some path $p \in Path_{s,v}$ where $p \notin Q_v$. $p$ is \textit{relevant due to optimality} with respect to $Q_v$ , if there exists a coefficient vector $ \lambda_{i} \in \Lambda $ where for each path $b \in Q_{v} $:
$$
Cost(p,\lambda_{i}) \leq Cost(b,\lambda_{i})
$$
\end{definition}

\begin{definition}[relevant due to dominance with respect to $Q_{v}$]\label{RelParetoDef}
Consider some node $v \in V$, a list of paths $Q_v \subseteq Path_{s,v}$ and some path $p \in Path_{s,v}$. $p$ is \textit{relevant due to dominance} with respect to $Q$ , if $p$ is Pareto non-dominated by each path $b \in Q_{v}$.
\end{definition}

We are ready to establish the definition of ``relevance''.
\definition[relevant with respect to Q]\label{nonDomDef} $p$ is $relevant$  with respect to $Q$ if it is either relevant due to optimality with respect to $Q_v$  or relevant due to dominance with respect to $Q$. 

\subsubsection{Is-Relevant procedure} \label{isRelAlgoSec}
In this subsection, we present the Is-Relevant procedure used by the IDAQ algorithm to determine whether a path $p \in Path_{s,v}$ is relevant with respect to a list of paths, $Q \subseteq Path_{s}$.\\
According to definition \ref{nonDomDef}, $p$ is relevant with respect to $Q$ if it meets either of the following conditions:
\begin{enumerate}
\item \label{ka} $p$ is relevant due to optimality with respect to $Q_{v}$ \footnote{Assuming the minimal cost for each  $\lambda_{i} \in \Lambda$ was calculated beforehand, it is easy to see that this can be verified in $O(K \cdot W)$.}.

 \item \label{pa} $p$ is relevant due to dominance with respect to $Q_{v}$ \footnote {Assuming the number of Pareto paths is given by $O(L)$, a non-dominance check can be verified in $O(L \cdot W)$. Note that the number of Pareto paths $L$ can be non-polynomial, however, in many practical applications $L$ could be relatively small \cite{paretoIsFeasible}.}
  \end{enumerate}
The following Is-Relevant procedure identifies which of the conditions (\ref{ka} or \ref{pa}) can be verified more efficiently. Eventually, this shall allow us to prove that, under reasonable assumptions, the IDAQ algorithm has lower time complexity than the Standard Algorithm.
\\

We begin by introducing an auxiliary definition followed by the Is-Relevant procedure pseudo-code.
\begin{definition}[$best_v$]\label{best}
Consider some node, say $v \in V$, and  a list of paths $Q \subseteq Path_{s}$. $best_v$ is define as an auxiliary list used by the Is-Relevant procedure, such that:
$$
best_v = \{best_v[1],best_v[2],...best_v[K]\}
$$
where for each $i=\{1,...,K\}$ and for any path $p \in Q_v$, the following holds:
$$best_v[i] \in Q_v \ \  \cap \ \ Cost(best_v[i],\lambda_i) \leq Cost(p,\lambda_i)$$
\end{definition}

\newpage
\begin{algorithmic}[1]\label{dominanceAlgo}
\Procedure{Is-Relevant}{$p$, Q, $\Lambda$}
\State $\text{K} \gets \textit{length of $\Lambda$}$
\State $\text{v} \gets \textit{p.Node}$
\If{$|Q_{v}|$ < $\frac{K}{W}$}
\State \textcolor{green}{ //Check for condition \ref{pa}}
\If{$p$ is Pareto Non-Dominated by $Q_{v}$}
 \State $P \gets $\textit{All paths that are dominated by $p$ in $Q$}
 \State $\textit{Remove P from Q}$
 \State $\textit{Return True}$
 \Else
 \State $\textit{Return False}$
 \EndIf
\Else 
\State \textcolor{green}{     //Check for condition \ref{ka}}
\State $\text{$p$.Costs} \gets \textit{$p$'s Costs For Each $
\lambda$}$    
\State $\text{$best_v$.Costs} \gets \textit{$best_v$ Costs For Each $\lambda$}$
\For{i=1 to K}
    \If{$p.Costs[i] \geq best_v[i].Costs[i]$}
    \State $p.Costs[i]=null$
    \Else
    \State $best_v[i]=p$
 \EndIf
\EndFor
\If{All $p.Costs$ are null}
 \State \textit{Return False}
 \Else
 \State $P \gets $ All path in $Q_v$ and not in $best_v$
 \State Remove $P$ from $Q$
 \State \textit{Return True}
 
  \EndIf
 \EndIf

\EndProcedure
\end{algorithmic}

\subsection{Adaptive Queue}
\label{adaptedQueue}
We turn to present the Adaptive Queue - a priority queue for determining the next path to be developed by the IDAQ algorithm. The Adaptive Queue, being a priority queue, is defined by the following operations: \textit{Push, Pop, Is-Empty}. Besides, the Adaptive Queue consists of one more operation, namely the \textit{Adapt} operation, used to sort the queue's paths according to a new coefficient vector. We begin with some auxiliary definitions.

\definition[$\Lambda\_index$] Index of coefficient vector in $\Lambda$ initialized with $1$.

\definition[Path Priority] \label{stateCost} Consider some node $v \in V$, path $p \in Path_{s,v}$ and $\Lambda\_index$. The $Priority$ of $p$, denoted by $p.Priority$, is equal to either null or $Cost(p,\lambda_{\Lambda\_index})$, as follows: 
\begin{enumerate}  
\item \textbf{null - } If there exists another path $b \in Path_{s,v}$ so that at least one of the following conditions holds:
\begin{enumerate}
    \item $b$ is either in the Adaptive Queue or was $popped$ from the Adaptive Queue, and: $$Cost(b,\lambda_{\Lambda\_index}) < Cost(p,\lambda_{\Lambda\_index})$$
    \item $b$ is in the Adaptive Queue, was $pushed$ to the Adaptive Queue before $p$ and: $$Cost(b,\lambda_{\Lambda\_index}) = Cost(p,\lambda_{\Lambda\_index})$$
    \item $b$ was popped from the Adaptive Queue and: $$Cost(b,\lambda_{\Lambda\_index}) = Cost(p,\lambda_{\Lambda\_index})$$
\end{enumerate}
\item \textbf{ Cost($p$,$\lambda_{\Lambda\_index}$) - } otherwise.
  \end{enumerate}

\definition[Priority Heap]\label{priorityHeap} A Fibonacci heap that contains each path in the Adaptive Queue with non-null Priority. The $Priority$ $Heap$ is managed by the Adaptive Queue operations (to be defined later in this section).\\

We are ready to define the Adaptive Queue operations:
\begin{description}

\item[Push($p$) -] 
Insert a Relevant (see definition \ref{nonDomDef}) path, say $p \in Path_{s,v}$ to the Adaptive Queue and perform the following actions:
\begin{enumerate}
\item Set $p.Priority$ according to definition \ref{stateCost}.
\item In case $p.Priority$ is not null:
\begin{enumerate}
    \item Set each path in $Q_v$ Priority to null and remove it, if necessary, from the Priority Heap
    \item Inset $p$ into the Priority Heap. 
\end{enumerate}
\end{enumerate}
  
\item[Is Empty() -] Check whether exists a path in the queue with not null Priority value. In other words, check if the Priory Heap is empty.

\item[Pop() -]\label{pop} Using the Priority Heap (see definition \ref{priorityHeap}), remove (and return) out of the Adaptive Queue a path with minimal not null Priority.

\item[Adapt(i) -] Perform the following actions:
\begin{enumerate}
\item Initialize the Priority Heap to a new empty heap.
\item Sets $\Lambda\_index=i$.
\item Calculate the Priority of each of the Adaptive Queue paths according to definition \ref{stateCost}.
\item Insert into the Priority each path with not null priority.
\end{enumerate}

\end{description}

\subsection{Build-Sets Procedure}\label{buildSetSec}
The IDAQ algorithm manages a list of optimal paths with respect to any coefficient vector in $\Lambda$. To solve $MOWSP$, IDAQ has to convert such list to a set of path for each coefficient vector as required by the $MOWSP$ definition (\ref{mowsp}).\\
Formally, the $Build-Sets$ procedure gets a list of paths, say $PL$, and returns for each coefficient vector $\lambda_i \in \Lambda$, a set of paths:
$$P_i=\{p_{v_1}, p_{v_2},...,p_{v_n}\}$$
where for each node $v_i \in V$: $p_{v_i} \in Path_{s,v_i}$ has minimum $Cost(p,\lambda_i)$ among all paths $p \in PL \cap Path_{s,vi}$.
\\
\newpage

\begin{algorithmic}[1]\label{buildSets}
\Procedure{Build-Sets}{$PL$, $\Lambda$,$G$}
\For{ i=1 to K}
\State $P_i \gets null$
\For{ each $v \in G.V$}
\State $paths \gets $ All paths in $PL$ ending in node $v$
\State $min \gets \infty$
\State $optimal\_path \gets null$ 
\For{ each $p \in paths$}
\If {$Cost(p,\lambda_i)<min$}
\State $optimal\_path \gets p$
\State $min \gets Cost(p,\lambda_i)$ 
\EndIf
\EndFor
\State $P_i \gets P_i \cup optimal\_path$
\EndFor
\State $optimal\_sets \gets optimal\_sets \cup \{P_i\}$
\EndFor
\State Return $optimal\_paths\_set$
\EndProcedure
\end{algorithmic}

\subsection{IDAQ algorithm}
\label{idaqPseudoSec}
We present the IDAQ algorithm, an efficient algorithm that solves MOWSP. We begin by introducing some auxiliary definitions followed by IDAQ's pseudo-code  (Algorithm \ref{idaqAlgo}).
\\

IDAQ uses an Adaptive Queue instance in order to manage a list of developed paths that is returned when the algorithm terminates (per the following definitions \ref{queue} and \ref{opPathsList}). 
\\

\begin{definition}[$queue$]\label{queue} An Adaptive Queue used by the IDAQ algorithm.
\end{definition}
\begin{definition}[$optimal\_paths$]\label{opPathsList}A list of paths that were popped from the $queue$ as paths with minimal non-null priority.
\end{definition}

Occasionally, we shall address the set of paths $Q$ defined as follows.

\begin{definition}[Discovered Paths]\label{Q} The set of paths that are in the $queue$ or the $optimal\_paths$ list. We denote the \textit{Discovered Paths} by $Q$.  
\end{definition} 
IDAQ uses an auxiliary list of paths in order to make sure that unnecessary paths will not be pushed to the $queue$. Such a list is defined as follows:

\begin{definition}[$pareto\_sample$]\label{singlePareto}A list of paths that is calculated in line 6 of IDAQ (Algorithm \ref{idaqAlgo}). Whenever the $pareto\_sample$ list contains exactly one path ending at some node $v \in V$, we denote it by $pareto\_sample[v]$.
\end{definition}
Indeed, we shall prove (Lemma \ref{paretoSample}) that the $pareto\_sample$ list contains precisely a single Pareto non-dominated path ending at each node.

\begin{definition}[iteration]\label{iteration}
A variable used by IDAQ. The value of the $iteration$ variable is initialized to $1$ (line 8 of Algorithm \ref{idaqAlgo}) and increased by $1$ each time the $queue$'s $Is-Empty$  operation \footnote{The $queue$ operations are defined in Section \ref{adaptedQueue}} returns true (line 14 of Algorithm \ref{idaqAlgo}). we shall denote the period where $iteration = i$ as "the $i$th iteration" or "iteration $i$".

\end{definition}
The IDAQ algorithm is specified as follows.
\begin{algorithm}
\caption{IDAQ}\label{idaqAlgo}
\begin{algorithmic}[1]
\Procedure{IDAQ}{G,s,$\Lambda$}
\State $E_{temp} \gets G.E$
\For{ each $e \in E_{temp}$}
\State $e.Cost \gets Cost(e, \lambda_{1})$
\EndFor
\State $G_{temp} \gets (G.V,E_{temp})$
\State $\textit{$pareto\_sample$} \gets Dijkstra(G_{temp}, s)$
\State $\textit{$optimal\_paths$} \gets empty\ list$
\State $\textit{iteration} \gets \text{1}$
\State $\textit{queue} \gets \text{Adaptive Queue (s, $\Lambda$,  iteration)}$
\While {$\textit{True}$}
\While {$queue$ is not empty} 
\If {$iteration$ > length of $\lambda$} 
\State \textit{Return $BuildSets(optimal\_paths,\Lambda,G)$}
\EndIf
\State $\textit{iteration} \gets \text{iteration+1}$
\State $\textit{queue.Adapt( iteration )}$
\EndWhile
\State $\textit{p} \gets \text{queue.Pop()}$
\State $\textit{Insert p to optimal\_paths} $
\For {$\textit{Each neighbor of p.Node}$}
\State $next\_p \gets \text{ Path(p , neighbor)}$
\State $pareto\_p \gets \textit{ pareto\_sample[p.Node]}$
\If {$pareto\_p$ Pareto Dominates $next\_p$ }
\State \textit{Continue}
\EndIf
\State $Q \gets queue \cup optimal\_paths$
\If {$Is-Relevant(next\_p, Q, \Lambda)$ } 
\State \textit{Push $next\_p$ to $queue$}
\EndIf

\EndFor
\EndWhile

\EndProcedure
\end{algorithmic}
\end{algorithm}
\newpage
\section{IDAQ Analysis}\label{IDAQCorrectSec}
In this section we prove the correctness of the IDAQ algorithm and analyze its time complexity.

\subsection{General definitions}

\begin{definition}[D]\label{D}  Average node degree:
$$D=\frac{|E|}{|V|}$$
\end{definition}

\begin{definition}[MaxDeg$_+$] Maximum in-degree of a node:
$$MaxDeg_+= max\{deg_+(v) | v\in V\}$$
\end{definition}

\begin{definition}[developed path]\label{devPath}
For the IDAQ algorithm, a {\it developed path} is a path that is popped from the $queue$ (Algorithm \ref{idaqAlgo} line 12-13), while for the Standard Algorithm, it is a path that is popped from the priority queue during a Dijkstra iteration.
\end{definition}

\begin{definition}[scanned path]
Consider a path $p \in Path_{s,v_{n}}$ where the list of nodes representing $p$ is given by: $\{s,v_1,v_2,...,v_{n-1},v_n\}$. $p$ is defined as a \textit{scanned path} if $p'$ is a developed path where:
$$
p' = \{s,v_1,v_2,...,v_{n-1}\}
$$
\end{definition}

\begin{definition}[developed node]\label{developedNode} A node that is at the end of a developed path. 
\end{definition}

\begin{definition}[ancestor]\label{ancestor} Consider a path $p \in Path_{s,v_{n}}$ where the list of nodes representing $p$ is given by: $\{s,v_1,v_2,...,v_n\}$. $p$'s {\it ancestors} are defined to be paths of the form: $\{s,v_1,v_2,...v_m\}$ where $m<n$.
\end{definition}
\begin{definition}[first optimal path]  Consider some coefficient vector, say $\lambda \in \Lambda$, and an optimal path with respect to $\lambda$, say $p \in Path_{s,v}$. $p$ is defined as a \textit{first optimal path} if at the time $p$ was scanned, neither the $queue$ nor the $optimal\_paths$ list contains another optimal $Path_{s,v}$ with respect to $\lambda$.
\end{definition}

\begin{definition}[first developed optimal path] Consider some coefficient vector, say $\lambda \in \Lambda$, and an optimal path with respect to $\lambda$, say $p \in Path_{s,v}$. $p$ is defined as a \textit{first developed optimal path} if at the time $p$ was developed, the $optimal\_paths$ list dose not contain another optimal $Path_{s,v}$ with respect to $\lambda$.
\end{definition}

\subsection{definitions given $l$}\label{complexityDef}
We proceed with some auxiliary definitions. 
Each of the following definitions depends on an input numerical integer value, namely $l$. The precise definition of $l$ shall be provided later;
intuitively, $l$  corresponds to the maximal number of Pareto paths for most nodes, which can be expected to be relatively small in practice \cite{paretoIsFeasible}.
\begin{definition}[$l$-node]\label{lnode}
A node that is at the end of at most $l$ different Pareto paths.
\end{definition}
 \begin{description}

\item[$V_L(l)$ - ] Set of all $l$-Nodes in $G$:

\begin{center}
$V_L(l) = \{v \in V  | v\ is\ an\ l$-Node$\}$ 
\end{center}

\item[$\alpha(l)$ - ] $l$-Nodes ratio. $$\alpha(l)= \frac{|V_L(l)|}{|V|}$$
\item[$E_L(l)$ -] Set of all outgoing edges from $l$-Nodes:
$$E_L(l) = \{ e\in E  | e.Start \in V_L(l) \} $$
\item[$\gamma(l)$ -] $l$-Node's outgoing edges ratio:
$$\gamma(l) = \frac{|E_L(l)|}{|E|} $$

\item[$D_L(l)$ -]  Average $l$-Node's degree:
$$D_L(l)=\frac{|E_L(l)|}{|V_L(l)|}$$

\item[$N_{L}(l)$ - ] The maximum number of paths non-dominated by a single Pareto path to any $l$-Node.

\item[$I(l)$ -] Difference between $K$ and $l$: $$I(l)=K-l$$

\end{description}

Next, we determine a specific value of $l$, namely $L$. Intuitively, $L$ is a small-enough number so that most of $G$ nodes are $L$-nodes.
\begin{definition}[L]\label{L} $L$ is defined as the solution of the following optimization problem:
$$
    L = \argminA_{l} l
$$
Subject to:
$$
\alpha(l) \geq  1- \frac{\log(|V|)}{|V|}
$$

\end{definition}
 In the following, whenever we do not state otherwise, $L$ is the default value of $l$ for each of the above defined definitions, e.g: $\alpha=\alpha(L)$, $V_{L}=V_{L}(L)$, etc.

\subsection{Assumptions}
\label{IDAQassumptionsSection}
We present several assumptions under which we can prove that IDAQ's time complexity is lower than that of the Standard Algorithm. Before stating each assumption, we provide some intuition that justifies it.
\\

In many classes of multi-objective optimization problems the number of objectives is relatively small, e.g. communication network routing problems often deal with just two objectives, e.g., throughput and delay.
\begin{assumption} \label{a0}
The number of objectives in a MOG is $O(1)$, i.e.,
$$W=O(1)$$
\end{assumption}


Moreover, many classes of networks are represented by sparse graphs, e.g. in road networks each intersection typically consists of the crossing of at most 4 roads. Thus:
\begin{assumption} \label{a2}
 The maximum number of incoming/outgoing edges from a single node is upper bounded by  $O(\log(|V|)$. In other words, $G$ is a sparse graph:
$$MaxDeg_-=MaxDeg_+=O(\log(|V|)$$
\end{assumption}

According to \cite{paretoIsFeasible}, in many practical applications, the number of different Pareto paths is often relatively small. Hence, a small number of input coefficient vectors (which is equal to the number of returned paths ending at each node), should produce paths with sufficient variety, i.e.:
\begin{assumption} 
\label{a3}
The number of input coefficient vectors in $\Lambda$, namely $K$ (equation \ref{lambdaEq}) is upper-bounded by  $O(\sqrt{|V|})$. 
$$K=O(\sqrt{|V|})$$
\end{assumption}

According to \cite{paretoIsFeasible}, in many practical applications, the various objectives are correlated. For example, in road networks, the shortest-distance route is among the fastest routes (albeit not necessarily the fastest). Accordingly, we assume that the number of paths that are non-dominated by a Pareto path is relatively small, i.e.:
\begin{assumption} 
\label{A5_1}
The maximum number of paths non-dominated by a single Pareto-non-dominated path to an $L$-Node (denoted by $N_{L}$), is upper-bounded by $O(L)$
$$
   N_{L}=O(L)
$$
\end{assumption}

Consider some $L$-Node $v \in V_{L}$. We assume that $N_{L}$ is smaller than $\frac{K}{W}$. Thus, IDAQ checks whether a $p \in Path_{s,v}$ Is-Relevant in an efficient matter by checking if $p$ is relevant by optimality (Definition \ref{relOpt}). \\
Note that $\frac{K}{W}=O(\sqrt{|V|}) $, while $N_{L}=O(L)=O(log(|V|))$, thus implying that the following assumption is acceptable:
\begin{assumption} 
\label{A5}
The maximum number of paths non-dominated by a single Pareto-non-dominated path to $L$-Node, $N_{L}$, is smaller or equal to $\frac{K}{W}$
$$
   N_{L}\leq{\frac{K}{W}}
$$
\end{assumption}

According to \cite{paretoIsFeasible}, the number of Pareto paths for most nodes is relatively small in practice. Intuitively, such a number should not be dramatically affected by the graph size; for example, in transportation networks, an "optimal road" from Boston to New-York would not pass through San Francisco. Hence:
\begin{assumption} 
\label{A6}
$$
   L=O(F(|V|) )
$$
where $F$ is a sub-logarithmic function, e.g.,
$F(|V|) = \sqrt{\log{|V|}}$, $F(|V|) = \log{\log{|V|}}$.

\end{assumption}

\subsection{IDAQ algorithm Correctness}
In this section, we prove that the IDAQ algorithm solves the MOWSP problem. In Theorem \ref{optStatelemma} we establish that, for any coefficient vector in $\Lambda$ and for each node $v \in V$, an optimal $Path_{s,v}$ is developed by the IDAQ algorithm. Next, based on Theorem \ref{optStatelemma}, in Theorem \ref{idaqcorrect} we shall prove that IDAQ solves the MOWSP problem. We begin by introducing and proving several auxiliary lemmas.

\begin{lemma}\label{optIsPareto}
An optimal $Path_{s,v}$ with respect to any coefficient vector in $\Lambda$ is Pareto non-dominated by any other $Path_{s,v}$.
\end{lemma}
\begin{proof}
Consider an optimal $Path_{s,v}$ with respect to $\lambda \in \Lambda$, namely $p_v$.
We assume, by way of contradiction, that $p_v$ is not Pareto-non-dominated, i.e., there exists some path $p_{opv} \in Path_{s,v}$ that dominates $p_v$. In other words, for each objective, say $i$:
$$\sum_{e \in p_v} w_e[i] \geq \sum_{e \in p_{opv}} w_e[i] $$
and for at least one objective, say $k$:
$$\sum_{e \in p_v} w_e[k] > \sum_{e \in p_{opv}} w_e[k] $$
Note that any coefficient vector in $\Lambda$ is composed by positive weights. Hence:
$$    Cost(p_v, \lambda) = 
    \sum_{e \in p_v}\sum_{j = 0 } ^{W}\lambda[j] \cdot w_e[j] > 
    $$
    $$
    \sum_{e \in p_{opv}}\sum_{j = 0 } ^{W}\lambda[j] \cdot w_e[j]  = Cost(P_{opv},\lambda)
$$
in contradiction to the optimality of $p_v$ with respect to $\lambda$. Hence, $p_v$ is Pareto non-dominated.
\end{proof}

\begin{lemma}\label{paretoSample}
The $pareto\_sample$ list contains exactly one Pareto non-dominated path ending at each node.
\end{lemma}
\begin{proof}
Recall that, in order to construct the $pareto\_sample$ paths list, an auxiliary graph, namely $G_{temp}$ is generated and on that graph we execute the Dijkstra algorithm.
\\

In lines 2-5 of IDAQ (Algorithm \ref{idaqAlgo}), the graph $G_{temp}$ is generated. Similarly to the first iteration of the Standard Algorithm, the nodes of $G_{temp}$ are precisely those of $G$, namely $V$. The edges of $G_{temp}$, namely $E_{temp}$, are equal to $E$ with the following addition: the cost value of each edge $e \in E_{temp}$, denoted by $e.Cost$, is determined by the $Cost$ function (line 4).
\\

The Dijkstra algorithm (executed in line 6) returns for each node $v \in V$ a path, namely $p_v \in Path_{s,v}$, which is the shortest path ending at $v$, with respect to the cost values $e.Cost$ for each edge $e \in E_{temp}$. 
\\
The output of the Dijkstra algorithm is denoted as $pareto\_sample$ where: 
$$pareto\_sample=\{p_{v_1}, p_{v_2},...,p_{v_n}\}$$

Recall that $G_{temp}$ has the same nodes and edges as $G$ hence, for any node $v$, each path $p \in Path_{s,v}$ is necessarily a path on both $G_{temp}$ and $G$.
\\
From the correctness of the Dijkstra algorithm we conclude that the path $p_{v}$ is such that:
$$
 p_{v}=\argminA_{p \in Path_{s,v}} f(p)
$$
where:
$$
    f(p) =\sum_{e \in p} e.Cost =\sum_{e \in p} Cost(e,\lambda_1) = Cost(p,\lambda_1)
$$
In other words, for each $p \in Path_{s,v}$ the following equation holds:    
$$
Cost(p_{v}, \lambda_1)  \leq Cost(p, \lambda_t)
$$
Hence, $p_v$ is by definition an optimal $Path_{s,v}$ with respect to $\lambda_1$ and in view of Lemma \ref{optIsPareto} we conclude that $p_v$ is indeed a Pareto non-dominated path.
\end{proof}

\begin{lemma}\label{optimalStateAncestor}
Consider an optimal path with respect to $\lambda_i\in \Lambda$, namely $p \in Path_{s,v_1}$. $p$'s ancestors are also optimal paths with respect to $\lambda_i$.
\end{lemma}
\begin{proof}
Assume by way of contradiction that there exists an ancestor path of $p$, say $b \in Path_{s,v_2}$, which is not optimal with respect to $\lambda_i$; in other words, there exists another path, namely $c \in Path_{s,v_2}$, which is optimal with respect to $\lambda_i$, i.e.:
$$
Cost(c,\lambda i) < Cost(b,\lambda i)
$$
Let us construct a new path to $v_1$, say  $d \in Path_{s,v_1}$, where: $d$ is the path $c$ followed by the path between $b$ and $p$. From the additivity of the cost function we conclude that:
$$
Cost(d,\lambda i) < Cost(a,\lambda i)
$$
in contradiction to the optimality of $p$ with respect to $\lambda_i$.
\end{proof}
In the following lemmas (\ref{firstOpsScan},\ref{optScanedInQueue}) we establish some properties that hold in case a first optimal path is scanned. Based on these, we will then show (Theorem \ref{optStatelemma}) that any first optimal path is scanned by the time that IDAQ terminates.
\begin{lemma}\label{firstOpsScan}
In case a first optimal path is scanned, it is pushed to the $queue$.
\end{lemma}
\begin{proof}
Consider the scan of a first optimal path, with respect to some coefficient vector $\lambda_i \in \Lambda$,say $p \in Path_{s,v}$.
\\
Recall that, by definition, at the time that a first optimal path is scanned, it is neither in the $queue$ nor in the $optima\_paths$ list.
\\
According to Lemma \ref{optIsPareto}, $p$ is non-dominated by another $Path_{s,v}$, in particular, $pareto\_sample[v]$. Therefore, the Is-relevant procedure is executed in order to determine whether $p$ is inserted to the $queue$.
\\
From the definition of a first optimal path we can conclude that, for each path $b \in Q_{v} $:
$$
Cost(p,\lambda_{i}) < Cost(b,\lambda_{i})
$$ 
Therefore, according to definition \ref{relOpt}, $p$ is relevant due to optimality with respect to $Q_v$. 
\\
From Lemma \ref{optIsPareto} we conclude that $p$ is Pareto non-dominated by any other $Path_{s,v}$, in particular $Q_v$. Hence, according to Definition \ref{RelParetoDef}, $p$ is relevant by dominance with respect to $Q_v$.
\\

Since $p$ is relevant by dominance with respect to $Q_v$ and it is relevant by optimality with respect to $Q_v$, it is identified as relevant with respect to $Q$ by the Is-Relevant procedure (\ref{isRelAlgoSec}) and pushed to the $queue$ (lines 24-25 of Algorithm \ref{idaqAlgo}), as required.
\end{proof}

In the following Lemma \ref{optScanedInQueue}, we shall prove that not only a scanned first optimal path is pushed to the $queue$ (as established by Lemma \ref{firstOpsScan}), but also it is not removed from the $queue$ unless the first developed optimal path is in the $optimal\_paths$ list.
\begin{lemma}\label{optScanedInQueue}
From the point in the execution of the algorithm that the first optimal path is scanned, either the first optimal path is in the $queue$  or the first developed optimal path is in the $optimal\_paths$ list.
\end{lemma}
\begin{proof}
Consider a first optimal path $p \in Path_{s,v}$ with respect to some coefficient vector $\lambda_i \in \Lambda$.
\\
According to Lemma \ref{firstOpsScan}, in case $p$ is scanned, it is pushed to the $queue$. Note that $p$ can only be removed from the $queue$ either by the $queue.pop$ operation or during the scan of some other path.
\\
In case the first developed optimal path is popped from the $queue$, it is inserted to the $optimal\_paths$ list (lines 16-17 of Algorithm \ref{idaqAlgo}), as required.
\\
Otherwise, assume by way of contradiction that $p$ was removed from the $queue$ during the scan of some other path, say $b \in Path_{s,v}$. One out of the following two cases must hold:
\begin{enumerate}
    \item \label{l5nd} $b$ dominates $p$.
    \item \label{l5null} All $p.Costs$ are null.
\end{enumerate}
Consider case \ref{l5nd}. From Lemma \ref{optIsPareto} we conclude that $p$ is Pareto non-dominated by any other $Path_{s,v}$, in particular $b$, hence this scenario is not possible.
\\
Consider case \ref{l5null}. Since $p$ is optimal with respect to $\lambda$, the following holds: 
$$Cost(p,\lambda_i) \leq Cost(b,\lambda_i)$$ Since $p$ has been inserted to the $queue$ before $b$, $p.Cost[i]$ is not null, this scenario is not possible.
\\
To conclude, we have established that none of the two scenarios is possible, in contradiction to our assumption, thus the lemma follows.

\end{proof}

\begin{theorem}\label{optStatelemma}
Consider any coefficient vector $\lambda_i \in \Lambda$ and any node $v \in V$. An optimal $Path_{s,v}$ with respect to $\lambda_i$ has been developed by the time that the IDAQ algorithm returns.
\end{theorem}
\begin{proof}
Assume by way of contradiction that there exists some node, say $v \in V$, and some coefficient vector $\lambda_i \in \Lambda$, so that no optimal $Path_{s,v}$ with respect to $\lambda_i$ has been developed by the end of the $i$th iteration (see Definition \ref{iteration}).
\\
Consider the first optimal path with respect to $\lambda_i$, namely $b \in Path_{s,v}$.
From Lemma \ref{optScanedInQueue} we can conclude that, in case that $b$ is scanned, it must be in the $queue$. Hence, by the time that the $i$th iteration ends, one of the following two cases must hold:
\begin{enumerate}
    \item \label{optScenarioPrevIter} $b$ is in the $queue$. 
    \item \label{BNever2} $b$ has yet to be pushed to the $queue$.
\end{enumerate}
Consider case \ref{optScenarioPrevIter}. $b.Priority$ must be null since the iteration ends only when the priority of all the paths in the $queue$ is null. Consider the priority calculation of path $b$ (either by the $queue.adapt$ operation or by the $queue.push$ operation). According to the path priority definition (\ref{stateCost}), since $b$ is a first optimal path, its priority can be set to null only if an optimal $Path_{s,v}$ was developed, in contradiction to our assumption.
\\

Turning to case \ref{BNever2}, consider the list of ordered nodes represented by path $b$:
$$b.Nodes=\{s,v_{b1}, v_{b2},... v\}$$

Let us define node $v_e$ as the last node in $b.Nodes$ for which an optimal $Path_{s,v_e}$ with respect to $\lambda_i$ has already been developed. We also denote by $e$ the first developed optimal $Path_{s,v_e}$ \footnote{Note that $e$ is not necessarily an ancestor of $b$ (see Definition \ref{ancestor})}.
\\

We note that $v_e$ and $e$ must exist since the first developed path, which contains only the start node $s$ (in other words: $v_e=s$  and $e.Nodes=\{s\}$) is a first developed optimal $Path_{s,v_e}$ with respect to $\lambda_i$.
In addition, $v_e \neq v$ since otherwise an optimal $Path_{s,v}$ with respect to $\lambda_i$ has been developed, contrary to our assumption.
\\

Let $v_f$ be the node after $v_e$ on the path $b$. Let  $f \in Path_{s,v_f}$ be the path $e$ followed by the node $v_f$.
\\

Consider the ancestor (see Definition \ref{ancestor}) of $b$ ending in $v_f$, namely $b_{f}$. According to Lemma \ref{optimalStateAncestor}, $b_{f}$ is optimal with respect to $\lambda_i$. From the additivity of $Cost$ and since $e$ is optimal with respect to $\lambda_i$ as well, we conclude that:
$$Cost(f,\lambda_i) = Cost(b_{f},\lambda_i).
$$
Hence, $f$ is optimal $Path_{s,v_f}$ with respect to $\lambda_i$.
\\
Recall that $v_e$ is followed by $v_f$ on the $b.Nodes$ list and that $v_e$ is the last node in the $b.Nodes$ list for which an optimal $Path_{s,v_e}$ with respect to $\lambda_i$ has already been developed, i.e.,an optimal $Path_{s,v_f}$ with respect to $\lambda_i$ has not been developed.
\\

Additionally, $f$ was scanned, since its ancestor, $e$, had been developed. Hence, either $f$ is the first optimal $Path_{s,v_f}$, or the first optimal $Path_{s,v_f}$ is already in the $queue$; in either case, the first optimal $Path_{s,v_f}$ with respect to $\lambda_i$ must be in the $queue$. According to Definition \ref{stateCost}, such a path does not have null priority. Hence, the iteration has not concluded, which is a contradiction.
\\\\
To conclude, we have established that none of the scenarios is possible, thus contradicting our assumption. Thus, after the $i$th iteration, an optimal $Path_{s,v}$ with respect to $\lambda_i$ has been developed.\\
The IDAQ algorithm returns only when the last iteration ends, hence the theorem follows. 

\end{proof} 

\begin{theorem}\label{idaqcorrect}
IDAQ solves the $MOWSP$ problem.
\end{theorem}
\begin{proof}
Recall that the IDAQ algorithm returns the following set: $$optimal\_sets=Build-Sets(optimal\_paths,\Lambda,G)$$
The Build-Sets procedure described in Section \ref{buildSetSec} returns for each coefficient vector $\lambda_i \in \Lambda$, a set of paths
$$P_i=\{p_{v_1}, p_{v_2},...,p_{v_n}\}$$
where for each node $v_i \in V$:
$$p_{v_i} \in Path_{s,v_i}$$

Assume by way of contradiction that there exists a set in $optimal\_sets$, say $P_j$, and some node, say $v \in V$, such that a $p_v \in P_j$ is not an optimal $Path_{s,v}$ with respect to $\lambda_j$. 
\\

We shall denote the first developed optimal $Path_{s,v}$ with respect to $\lambda_j$ by $op_v$. Note that, according to Theorem \ref{optStatelemma}, an optimal $Path_{s,v}$ with respect to $\lambda_j$ has been developed by the time that the IDAQ algorithm returns, i.e., $op_v$ has been developed by that time. Recall that any path that has been developed is inserted into the $optimal\_paths$ list (lines 16-17 of Algorithm \ref{idaqAlgo}), hence $op_v$ is in the $optimal\_paths$ list. In other words,
$$op_v \in optimal\_paths\cap Path_{s,v}$$

According to the description of the Build-Sets procedure (Section \ref{buildSetSec}), for any paths set in $optimal\_sets$, say $P_j$, and for any node in $V$, say $v$, the following holds: $p_v\in P_j$, achieves the minimum value of $Cost(p,\lambda_j)$ among all paths $p \in optimal\_paths\cap Path_{s,v}$. \\
Since $op_v \in optimal\_paths\cap Path_{s,v}$ the following holds:
$$Cost(p_v,\lambda_j) \leq Cost(op_v,\lambda_j)$$
According to Definition \ref{opDef}, since $op_v$ is optimal with respect to $\lambda_j$, for any $p \in Path_{s,v}$ the following holds:
$$Cost(op_v,\lambda_j) \leq Cost(p,\lambda_j)$$
Hence,
$$Cost(p_v,\lambda_j) \leq Cost(op_v,\lambda_j) \leq Cost(p,\lambda_j)$$
Thus, $p_v$ is by definition an optimal $Path_{s,v}$ with respect to $\lambda_j$, which is a contradiction, hence establishing the theorem.
\end{proof}

\subsection{IDAQ Time Complexity Analysis}\label{IDAQComplexityAnalisys}
We turn to analyze the time complexity of the IDAQ algorithm. We divide IDAQ operations into four parts and analyze each separately.
\begin{enumerate}
    \item \label{Part1} Developing paths (theorem \ref{developTime}).
    \item \label{Part2} Scanning paths ( theorem \ref{P9}).
    \item \label{Part3} Adapting the $queue$ between iterations (theorem \ref{AdaptThoreme}). 
    \item \label{Part4} Executing the Build-Sets procedure. 
\end{enumerate}
The overall time complexity of IDAQ shall be established in Theorem \ref{IDAQcomplexity}.
\\


\subsubsection{Complexity Analysis Part \ref{Part1} - Developing Paths}\label{part1Analysis}
In this subsection we establish the time complexity analysis of developing paths in IDAQ (Part \ref{Part1}). The main motivation for IDAQ is to avoid developing paths that have already been developed in previous iterations, therefore we expect to get here a lower bound on the number of developed paths (and, as a result, a lower time complexity) than for the Standard Algorithm.
\\

First, we calculate the time complexity of developing paths (Theorem \ref{developTime}). We begin by introducing and proving several auxiliary lemmas.\\

\begin{lemma}\label{optimalStateNotNull}
Consider a first optimal path with respect to $\lambda_i\in \Lambda$, say $p \in Path_{s,v}$. In case $p$ is pushed to the queue during the $i$th iteration, the priority of each of the $queue$'s paths ending at node $v$ is set to null.
\end{lemma}
\begin{proof}
Consider any path $b \in Path_{s,v}$.
Note that, since $p$ is first optimal with respect to $\lambda$, at the time $p$ is pushed to the queue, the following holds:
$$Cost(p,\lambda_i)<Cost(b,\lambda_i)$$
Hence, by the definition of path priority (Definition \ref{stateCost}), $b$'s priority is set to null, as required.

\end{proof}

In the following lemma we establish that the IDAQ algorithm does not develop non-optimal paths. This shall allow us to reach an upper bound on the number of developed paths.
\begin{lemma}\label{optStatetheorem}
For $i \in \{1 ,2,...,K\}$, a developed path in iteration $i$ is optimal with respect to $\lambda_i$.
\end{lemma}
\begin{proof} 
Consider some iteration (see definition \ref{iteration}), say $i$, and a path $p \in Path_{s,v}$ that was developed during the $i$th iteration.
Assume by way of contradiction that $p$ is not optimal with respect to $\lambda$. In other words, there exists another path $b \in Path_{s,v}$ such that $b$ is first optimal with respect to $\lambda_i \in \Lambda$, i.e.:
$$Cost(b,\lambda i)<Cost(p,\lambda_i)$$
One of the following cases must hold:
\begin{enumerate}
    \item \label{optScenarioPrevIter2} $b$ and $p$ were pushed to the $queue$ during a previous iteration.
    \item \label{optScenarioPrevBCurrA} $b$ was pushed to the $queue$ before $p$; $p$ was pushed to the $queue$ during iteration $i$.
    \item \label{optScenarioPrevACurrB} $p$ was pushed to the $queue$ before $b$; $b$ was pushed to the $queue$ during iteration $i$.
    \item \label{optScenarioBNever} $b$ has yet to be pushed to the $queue$.
\end{enumerate}

Consider case \ref{optScenarioPrevIter2}. At the beginning of the $i$th iteration, $p$ is in the $queue$ and $b$ is in either the $queue$ or the $optimal\_paths$ list. We examine the priority calculation of $p$ by the $Adapt$ operation at the beginning of the $i$th iteration. According to Definition \ref{stateCost}, $p$'s priority is set to null, thus implying that it could not have been developed during the $i$th iteration.
\\

Consider case \ref{optScenarioPrevBCurrA}. According to Lemma \ref{optScanedInQueue}, during the Push of $p$ to the $queue$, either $b$ is in the $queue$ or the first developed optimal $Path_{s,v}$ is in the $optimal\_paths$ list. We examine the priority calculation of $p$ by the $Push$ operation. According to Definition \ref{stateCost}, $p$'s priority is set to null, thus implying that it could not have been developed during the $i$th iteration.
\\

Consider case \ref{optScenarioPrevACurrB}. $b$ is first optimal with respect to $\lambda_i$, therefore, at the time that $b$ is being pushed to the $queue$ (during iteration $i$), $b$.Priority is not null. According to Lemma \ref{optimalStateNotNull}, the priority of each of the $queue$'s paths ending at node $v$ is set to null, in particular, the priority of $p$, thus implying that $p$ could not have been developed during the $i$th iteration.
\\
Consider case \ref{optScenarioBNever}. Note that, Since $b$ is first optimal,
$$Cost(b,\lambda_i) < Cost(p,\lambda_i)$$

Consider the list of ordered nodes represented by path $b$:
$$b.Nodes=\{s,v_{b1}, v_{b2},... v\}$$

Let us define node $v_e$ as the last node in $b.Nodes$ for which an optimal $Path_{s,v_e}$ with respect to $\lambda_i$ has already been developed. We also denote by $e$ the first developed optimal $Path_{s,v_e}$ \footnote{Note that $e$ is not necessarily an ancestor  of $b$ (see Definition \ref{ancestor})}.
\\

We note that $v_e$ and $e$ must exist since the first developed path, which contains only the start node $s$ (in other words: $v_e=s$  and $e.Nodes=\{s\}$) is a first developed optimal $Path_{s,v_e}$ with respect to $\lambda_i$.
In addition, $v_e \neq v$ since otherwise an optimal $Path_{s,v}$ with respect to $\lambda_i$ has been developed, contrary to our assumption.
\\

Let $v_f$ be the node after $v_e$ on the path $b$. Let  $f \in Path_{s,v_f}$: $f$ be the path $e$ followed by the node $v_f$. 
\\
Consider the ancestor (see Definition \ref{ancestor}) of $b$ ending in $v_f$, namely $b_{f}$. According to Lemma \ref{optimalStateAncestor}, $b_{f}$ is optimal with respect to $\lambda_i$. From the additivity of $Cost$ and since $e$ is optimal with respect to $\lambda_i$ as well, we conclude that:
$$Cost(f,\lambda_i) = Cost(b_{f},\lambda_i)
$$
Hence, $f$ is optimal $Path_{s,v_f}$ with respect to $\lambda_i$.
\\
Recall that $v_e$ is followed by $v_f$ on the $b.Nodes$ list and that $v_e$ is the last node in the $b.Nodes$ list for which an optimal $Path_{s,v_e}$ with respect to $\lambda_i$ has already been developed, i.e., an optimal $Path_{s,v_f}$ with respect to $\lambda_i$ has not been developed.
\\

In addition, $f$ was scanned, since its ancestor, $e$, had been developed. Hence, either $f$ is first optimal $Path_{s,v_f}$, or the first optimal $Path_{s,v_f}$ is already in the $queue$; in either case, the first optimal $Path_{s,v_f}$ with respect to $\lambda_i$ must be in the $queue$. We shall denote such a path by $g$. According to Definition \ref{stateCost}, $g.Priority$ is not null. Additionally, we have:
$$g.Priority = Cost(g,\lambda_i)=Cost(f,\lambda_i) =$$
$$Cost(b_{f},\lambda_i)\leq Cost(b,\lambda_i) < Cost(p,\lambda_i)$$
i.e., $g.Priority < p.Priority$, in contradiction to $p$ being a path with minimal not null priority in the $queue$.
\\\\
To conclude, we have established that none of the scenarios is possible, in contradiction to our assumption.
\end{proof}

\begin{lemma}
\label{l1}
In each iteration, IDAQ develops each node at most once. 
\end{lemma}
\begin{proof}
Assume by a way of contradiction that there exists a node that was developed more than once, say during iteration $i$. Consider $b$ and $p$ as the first two $Path_{s,v}$s that were developed during the $i$th iteration. Without loss of generality, we assume that $b$ was pushed to the $queue$ before $p$.
\\

According to Lemma \ref{optStatetheorem}, both $p$ and $b$ have minimal cost with respect to $\lambda_i$ among all paths in $Path_{s,v}$, i.e.:
$$
Cost(p,\lambda_i) = Cost(b,\lambda_i).
$$
One of the following two cases must hold:
\begin{enumerate}
\item \label{6_c1} $b$ and $p$ were pushed to the $queue$ in a previous iteration.
\item \label{6_c2} $b$ was pushed to the $queue$ before $p$, while $p$ was pushed to the $queue$ during the $i$th iteration.
\end{enumerate}

Consider case \ref{6_c1}. At the beginning of the $i$th iteration, $p$ and $b$ are in the $queue$. We examine the priority calculation of $p$ by the $Adapt$ operation at the beginning of the $i$th iteration. According to definition \ref{stateCost}, $p$'s priority is set to null, thus implying that it could not have been developed during the $i$th iteration.
\\

Consider case \ref{6_c2}. Namely $b$ is in the $queue$ during the scan of $p$. Therefore, at the time that $p$ is being pushed to the $queue$, $p$.Priority is set to null (Definition \ref{stateCost}). According to the definition of the $Push$ operation, the priority of each of the $queue$'s paths ending at node $v$ is set to null, in particular, the priority of $b$, thus implying that $b$ could not have been developed during the $i$th iteration.
\\

To conclude, we have established that none of the scenarios is possible, in contradiction to our assumption.

\end{proof}
\begin{lemma}
\label{l2}
IDAQ Develops any path \footnote{Out of the set of all possible paths in $Path_{s}$} no more than once. 
\end{lemma}
\begin{proof}
Consider some node $v \in V$. By way of contradiction, let us assume that there exists a path $b \in Path_{s,v}$, which was pushed to the $queue$ after path $ p \in Path_{s,v}$ so that both $p$ and $b$ represent the same path ending at node $v$ (i.e. $p=b$) and have been developed at the time that the IDAQ algorithm returns.
\\

Consider the scan of $b$. Since $p$ was developed and $p$ was pushed to the $queue$ before $b$, $p$ is either in the $queue$ or in the $optimal\_paths$ list. Recall that $Q$ is defined as the group of paths that are either in the $queue$ or in the $optimal\_paths$ list (Definition \ref{Q}), i.e., $p \in Q$. According to Definition \ref{nonDomDef}, $b$ is not Relevant with respect to $Q$, hence, $p$ could not have been inserted to the $queue$ and cannot be developed, in contradiction to our assumption.
\end{proof}

\begin{lemma} \label{a1} The number of nodes that are not $L$-Nodes (see Definition \ref{lnode}) is bounded by $O(\log(|V|)$:
$$(1-\alpha) \cdot |V|=O(\log(|V|)$$
\end{lemma}
\begin{proof}
According to Definition \ref{L}, $L$ is determined so that:
$$
\alpha \geq  1- \frac{\log(|V|)}{|V|} \iff 
$$
$$
(1-\alpha) \cdot |V| \leq \log(|V|)
$$
i.e:
$$
(1-\alpha) \cdot |V| = O(\log(|V|))
$$
thus the lemma follows.

\end{proof}

\begin{lemma}\label{lNodOpt}
Consider any $L$-node, say $v \in V_L$.  There are at most $L$ different $Path_{s,v}$'s that are optimal with respect to any coefficient vector in $\Lambda$.
\end{lemma}
\begin{proof}
In Lemma \ref{optIsPareto} we have established that each optimal path is Pareto non-dominated. Consider some $L$-node, say $v \in V_L$. According to Definition \ref{lnode}, $v$ is at the end of at most $L$ different Pareto paths. Therefore, the number of optimal paths with respect to any coefficient vector is upper-bounded by $L$, as required.
\end{proof}

\begin{lemma}\label{devLNode}
The number of developed paths that end at an $L$-node is upper-bounded by:
$$|D_{L-nodes}| = O(|V_L| \cdot L) =O(|V| \cdot \alpha \cdot L )$$
\end{lemma}
\begin{proof}

By definition (see Section \ref{complexityDef}):
$$\alpha= \frac{|V_L|}{|V|}$$
Hence, the number of $L$-nodes is given by: $$|V_L|=|V|  \cdot  \alpha$$
Recall that:
\begin{enumerate}
    \item Due to Lemma \ref{lNodOpt}, for any $L$-node, say $v \in V_L$, there are at most $L$ different $Path_{s,v}$'s that are optimal with respect to any coefficient vector in $\Lambda$.
    \item  Due to Lemma \ref{optStatetheorem}, each path developed by IDAQ is optimal.
    \item Due to Lemma \ref{l2}, during an execution of the IDAQ algorithm, a path is developed no more than once.
\end{enumerate}
Therefore, the number of developed paths that end at an $L$-node is upper-bounded by:
$$|D_{L-nodes}| = O(|V_L| \cdot L) =O(|V| \cdot \alpha \cdot L )$$,
as required.
\end{proof}

\begin{lemma}\label{devNotLNode}
The number of developed paths ending at a node that is not an $L$-node is upper-bounded by:
$$|D_{Not-L-nodes}| = O(|V 	\setminus V_L| \cdot K)=O(|V| \cdot (1-\alpha) \cdot K )$$
\end{lemma}
\begin{proof}
By definition (see Section \ref{complexityDef}):
$$\alpha= \frac{|V_L|}{|V|}  \iff |V_L|=|V|  \cdot  \alpha$$
$$ \iff |V| - |V_L|=|V|  \cdot  (1-\alpha)$$
Since $V_L \subseteq V$, we conclude that the number of non-$L$-nodes is given by: 
$$|V \setminus V_L|=|V|  \cdot  (1-\alpha)$$ 

In Lemma \ref{l1} we have established that, at each iteration, IDAQ develops each node at most once. Recall that the number of iterations is given by $K$, i.e., the number of developed paths ending at any node is upper-bounded by $K$. Hence the number of developed paths ending at a node that is not an $L$-node is upper-bounded by:
$$|D_{Not-L-nodes}| =  O(|V \setminus V_L| \cdot K)=O(|V| \cdot (1-\alpha) \cdot K )$$,
as required.
\end{proof}

Finally, we employ lemmas \ref{devLNode} and \ref{devNotLNode} to get the total time complexity of developing nodes in IDAQ, which is presented by the following theorem.
\\
\begin{theorem}
\label{developTime}
The time complexity of developing paths in IDAQ (part \ref{Part1}) is \footnote{Recall that, according to Assumption \ref{A6}, $F(|V|)$ is a sub-logarithmic fun}: 
\\
$$
O(L \cdot |V| \cdot \log(|V|)) =
O(F(|V|) \cdot |V| \cdot \log(|V|))
$$
\end{theorem}
\begin{proof}\

In Lemma \ref{devLNode}, we have established that the number of developed paths that end at an \textbf{$L$-node} is upper-bounded by:
$$|D_{L-nodes}| =  O(|V| \cdot \alpha \cdot L )$$
In Lemma \ref{devNotLNode}, we have demonstrated that the number of developed paths ending at a \textbf{node that is not an $L$-node} is upper-bounded by:
$$|D_{Not-L-nodes}| = O(|V| \cdot (1-\alpha) \cdot K )$$
Obviously, each path in $Path_{s}$ must end either at an $L$-node or at a node that is not an $L$-node, hence the number of developed paths is upper-bounded by $|P_{Developed}|$, where:
$$
 |P_{Developed}|=|D_{Not-L-nodes}|+|D_{L-nodes}|= $$
$$
 \underbrace{|V| \cdot (1-\alpha) \cdot K }_\text{(I)}+ \underbrace{|V| \cdot \alpha \cdot L}_\text{(II)}=(I)+(II)
$$
and:
$$
(I) =|V| \cdot (1-\alpha) \cdot K
$$
$$
(II)=|V| \cdot \alpha \cdot L
$$

Using Lemma \ref{a1} and Assumption \ref{a3} we conclude that:
$$
(I)=K \cdot O(\log(|V|)) = O(\sqrt{|V|} \cdot \log(|V|)).
$$
Since $\alpha < 1$ and due to Assumption \ref{A6} we can conclude that:
$$
(II) = O(L \cdot |V|) = O(F(|V|) \cdot |V|)
$$
We thus get:
$$
|P_{Developed}|=O(L \cdot |V|)=O(F(|V|) \cdot |V|).
$$
Note that each developed path has to be $popped$ out of the $queue$ as a path with minimal non-null priority.
\\
From Definition \ref{stateCost} we conclude that, for each node, say $v \in V$, the  $queue$ contains no more then a single $Path_{s,v}$ with non-null priority. Hence, the total number of paths with non-null priority in the $queue$, which is equal to  the number of paths in the Priority Heap (definition \ref{priorityHeap}), is upper-bounded by $|V|$. 
\\
The $pop$ operation (described in Section \ref{adaptedQueue}) consists of finding a path with minimal priority in the Priority Heap (which is an instance of the Fibonacci Heap data structure).  Therefore, a single $pop$ is done in $O(\log(|V|))$ \cite{fibQueueDijkstraComplexity}. 
\\

We thus conclude that the time complexity of developing paths in IDAQ is:
$$
O(|P_{Developed}| \cdot \log(|V|))=
$$$$
O(L \cdot |V| \cdot \log(|V|)=
O(F(|V|) \cdot \log{|V|} \cdot |V|)
$$
\end{proof}
\subsubsection{Complexity Analysis Part \ref{Part2} - Scanning paths}\label{part2Analysis}
We proceed to analyze the time complexity of scanning paths in IDAQ (part \ref{Part2}).  IDAQ checks whether a scanned path $p \in Path_s$ is potentially optimal (by the Is-Relevant procedure defined in Section \ref{isRelAlgoSec}). In that case, IDAQ pushes $p$ to the $queue$, which make it unnecessary for $p$'s ancestor to be developed again at a future iteration. We shall show that, despite the fact that a single IDAQ scan is more time-consuming than a scan of the Standard Algorithm, the total time complexity of the scans is lower for IDAQ.
\\

We begin by establishing the following lemma.
\begin{lemma}\label{eq8}
The following equation holds:
\begin{equation}\label{OL}
    I \cdot (1-\gamma)= O(L)
\end{equation}
\end{lemma}
\begin{proof}
Since $I=K-L$, we have:
$$
I \cdot (1-\gamma)=O(L) \iff
$$
$$
O(K) \cdot (1-\gamma)=O(L) \iff
$$
$$
1-\gamma= \frac{O(L)}{O(K)} \iff
$$
$$
    |E| \cdot (1-\gamma)=\frac{O(L \cdot |E|)}{O(K)}
$$
To sum up, Equation \ref{OL} holds if the following Equation \ref{NLG} holds:
\begin{equation}
    \label{NLG}
    |E| \cdot (1-\gamma)=\frac{O(L \cdot |E|)}{O(K)}
\end{equation}

Since $E=|V| \cdot D=O(|V| \cdot MaxDeg_{+})$, where $MaxDeg_{+}$ is defined as the maximum number of edges emanating from a single node, we conclude that:
$$
\frac{O(L \cdot |E|)}{O(K)}=\frac{O(L \cdot |V| \cdot MaxDeg_{+})}{O(K)}
$$
From Assumption \ref{a3} we conclude that:
$$
=\frac{O(L) \cdot O(|V| \cdot MaxDeg_{+}))}{O(\sqrt{|V|})}
$$
$$
=\frac{O(L) \cdot O(|V|) \cdot O(MaxDeg_{+})))}{O(\sqrt{|V|})}
$$
$$
=O(L) \cdot O(\sqrt{|V|})) \cdot O(MaxDeg_{+})
$$
From assumptions \ref{a2} and \ref{A6} we conclude that:
$$
=O(F(|V|) \cdot O(\sqrt{|V|})) \cdot O(\log(|V|)))
$$
$$
=O(F(|V|) \cdot \sqrt{|V|}) \cdot \log(|V|))
$$
$$
=O(\sqrt{|V|} \cdot \log(|V|) \cdot F(|V|)) 
$$
Hence, Equation \ref{NLG}\ holds if the following Equation \ref{EGV} holds:
\begin{equation}
    \label{EGV}
    |E| \cdot (1-\gamma)=O(\sqrt{|V|} \cdot \log(|V|) \cdot F(|V|)) 
\end{equation}
Note that the number of edges emanating from nodes that are not $L$-nodes is upper-bounded by the number of such nodes times the maximum nodal degree, namely:
$$
|E| \cdot (1-\gamma)=(|V|-|V_{L}|) \cdot O(MaxDeg_{+})
$$
$$
=(1-\alpha) \cdot |V| \cdot O(MaxDeg_{+})=O((1-\alpha) \cdot |V| \cdot MaxDeg_{+})
$$
Using Lemma \ref{a1} and Assumption \ref{a2} we conclude that:
$$
|E| \cdot (1-\gamma)=O(\log(|V|) \cdot \log(|V|))=O(\log^2(|V|))
$$
Hence, Equation \ref{EGV} holds, thus implying that Equation \ref{OL} holds, and the lemma follows.
\end{proof}
Next, we employ lemmas \ref{devLNode}, \ref{devNotLNode} and \ref{eq8} to establish an upper-bound on the total number scans performed during an execution of the IDAQ algorithm.

\begin{theorem}
 \label{P6}
 An upper-bound of the number of scans performed during an execution of the IDAQ algorithm is given by $|P_{Scanned}|$, where:
 \begin{equation}
     |P_{Scanned}| = O(|E| \cdot L).
 \end{equation}
\end{theorem}
\begin{proof}
Note that a path that was developed causes all of its neighbors to be scanned. Hence, the  number of scanned paths equals the number of outgoing edges from each developed path.\\

Recall that $E_{L}$ is defined as the set of all outgoing edges from $V_L$ (see Section \ref{complexityDef}). In Lemma \ref{devLNode}, we have established that the number of developed paths ending at an $L$-node is upper-bounded by:
$$|D_{L-nodes}| = O(|V_L| \cdot L) $$
Hence, the number of scans of paths that are neighbors of an $L$-nodes is given by:
$$
|S_{L-nodes}| = O(  |E_L| \cdot L )
$$

Recall that the set of nodes that are not $L$-nodes is given by $V \setminus V_L$. Since the set of edges is given by $E$, and each node is either an $L$-node or a non-$L$-node, the set of all outgoing edges from a $V \setminus V_L$ is given by $ E \setminus E_{L} $. In Lemma \ref{devNotLNode} we have established that the number of developed paths ending at a node that is not an $L$-node is upper-bounded by:
$$|D_{Not-L-nodes}| = O(|V \setminus V_L| \cdot K)$$
Hence, the number of scans of paths that are neighbors of non-$L$-nodes is given by:
$$
|S_{Not-L-nodes}| = O(|E \setminus E_L| \cdot K)
$$

Thus, the total number of scans is given by $|P_{Scanned}|$, where:
$$
|P_{Scanned}| = |S_{L-nodes}| + |S_{Not-L-nodes}| 
$$
$$
=   O(|E_L| \cdot L + |E \setminus E_L| \cdot K)
$$
$$
=L \cdot |E_L| + (|E|-|E_L|) \cdot K 
$$
$$
= |E| \cdot (L \cdot \frac{|E_L|}{|E|} + K \cdot (1-\frac{|E_L|}{|E|}))
$$
From the definition of $\gamma$ (Section \ref{complexityDef}), we conclude that:
$$
|P_{Scanned}| =|E| \cdot (L \cdot \gamma + K \cdot (1-\gamma))
$$
Since $K=L+I$:
$$
|P_{Scanned}| =|E| \cdot (L \cdot \gamma + L + I -L \cdot \gamma - I \cdot \gamma )
$$
$$
=|E| \cdot ( L + I \cdot (1 -  \gamma ))
$$
In Lemma \ref{eq8} we have established that:
$$
 I \cdot (1-\gamma)= O(L),
$$
hence, 
$$
|P_{Scanned}|= |E| \cdot (L+O(L))=|E| \cdot O(L)= O(|E| \cdot L),
$$
as required.

\end{proof}
We proceed to calculate (in Lemma \ref{P5}) the time complexity of scanning a single path ending at an $L$-node. We begin by establishing an auxiliary lemma.
\\

For the following lemmas we consider the set of paths $Q$ (definition \ref{Q}), and the subset of paths $Q_v \subseteq Q$ (definition \ref{Qv}). 

\begin{lemma}
\label{P4}
For any $L$-node $v \in V_L$, $Q$ contains no more than $N_{L}$ paths ending at $v$. 
$$
|Q_v| \leq N_L
$$
\end{lemma}
\begin{proof}
According to Lemma \ref{paretoSample}, after initialization, the $pareto\_sample$ list contains a Pareto non-dominated path for each $L$-node. Consider such a Pareto non-dominated path $p \in Path_{s,v}$ where $v \in V_{L}$. According to assumption \ref{A5_1}, there exist at most $N_{L}$ paths ending at $v$ that are non-dominated by $p$.\\
A path that is pushed to the $queue$ must be non-dominated by $p$ (lines 21-22 of Algorithm \ref{idaqAlgo}), thus implying that there exist at most $N_{L}$ paths ending at node $v$ that can be pushed to the $queue$ during an execution of the IDAQ algorithm. Only paths that were pushed to the $queue$ can be inserted to the $optimal\_paths$ list, thus the lemma follows.
\end{proof}

\begin{lemma}
\label{P5}
Scanning a single path that ends at an $L$-node is done in $O(L)$.
\end{lemma}
\begin{proof}
Consider some $L$-node $v \in V_L$ and a path $p \in Path_{s,v}$. Recall that, during a scan of a path $p$, the IDAQ algorithm checks whether $p$ Is-Relevant with respect to $Q$  (defined in Section  \ref{isRelAlgoSec}), in which case IDAQ pushes $p$ to the $queue$ (line 25 of Algorithm \ref{idaqAlgo}). 
\\

Note that, by Lemma \ref{P4}, $Q$ contains no more then $N_{L}$ paths ending at any $L$-node, i.e., 
 $$|Q_v|<N_L$$
In addition, due to Assumption \ref{A5_1}, $$N_{L}= O(L).$$
Hence,
$$|Q_v|=O(L).$$
We shall calculate the time complexity of scanning $p$.
The operations performed during the scan of $p$ are as follows:
\begin{enumerate}
\item \label{action3} Checking whether $p$ is non-dominated by the paths $Q_v$ (Is-Relevant procedure).

\item \label{action1} Checking whether $p$ dominates paths in $Q_v$ and removing them (Is-Relevant procedure). 

\item  \label{action2} Updating the priority of $queue$'s paths ending at node $v$ (Push operation).

\item \label{action4} Inserting $p$ to the $queue$ (Push operation).
\end{enumerate}
Consider operation \ref{action3}. A single Pareto dominance check can be done in $O(w)$. i.e, the time complexity of checking whether $p$ is Pareto non-dominated by each path in $Q_v$ is upper bounded by:
$$O(N_{L} \cdot W)=O(L \cdot W).$$
Consider operation \ref{action1}, In a similar manner to operation \ref{action3} above, this can be done in $O(L \cdot W)$.
\\
Consider operation \ref{action2}. It involves checking if $p$'s cost is lower in comparison to each of the $queue$ paths ending at $v$. i.e, this can be done in $O(L)$.
\\
Consider operation \ref{action4}. $p$ is added to the $queue$'s data structure, which is done in $O(1)$. In case $p$ is inserted to the $queue$'s Priority Heap (Definition \ref{priorityHeap}), it has to be inserted to a Fibonacci Heap, which is also executed in $O(1)$ \cite{fibQueueDijkstraComplexity}.
\\
By Assumption 1,
$$W = O(1)$$
We thus get, that a scan of an $L$-node is done in $O(L)$ as required.
\end{proof}

Next, we calculate the time complexity of scanning paths that end at nodes which are not $L$-nodes. We begin by establishing several auxiliary lemmas.

\begin{lemma}\label{l4_1}
$Q$ contains no more then $K$ paths ending at any node $v \in V$.
$$
|Q_v| \leq K
$$
\end{lemma}
\begin{proof}
Consider some node $v \in V$. According to Lemma \ref{P4}, in case $v$ is an $L$-node:
$$|Q_v| \leq N_L \leq \frac{K}{W}\leq K.$$
Otherwise, the Is-Relevant procedure updates $Q_v$ to contain only paths that have the best cost for at least a single coefficient vector $\lambda \in \lambda$. Since $|\Lambda| = K$, there exist at most $K$ different paths in $Q_v$.
\end{proof}

\begin{lemma} \label{l4}
A single IDAQ scan of a path that ends at a node that is not an $L$-nodes is done in $O(K)$.
\end{lemma}
\begin{proof}
Consider some node that is not an $L$-node, $v \notin V_{L}$ and path $p \in Path_{s,v}$. During a scan of  $p$, the IDAQ algorithm checks whether $p$ Is-Relevant with respect to $Q$  (defined in Section  \ref{isRelAlgoSec}), in which case IDAQ pushes $p$ to the $queue$ (line 25 of Algorithm \ref{idaqAlgo}). 
\\

In case the number of paths in $Q_v$ is smaller than $\frac{K}{W}$, the actions involved in scanning $p$ are identical as if $v$ were an $L$-node. In other words, it can be done in $O(L*w)$ (Lemma \ref{P5}). Otherwise, the actions performed during the scan of $p$ are as follows:
\begin{enumerate}

\item \label{action11} Calculating $p$'s cost for each iteration (Is-Relevant procedure).

\item \label{action12} Finding a path with minimal cost for each iteration in $Q_v$ (Is-Relevant procedure).

\item \label{action13}Updating $p$'s cost to null for each iteration $i$ where $p.Cost[i]>Best[i].Cost[i]$ (Is-Relevant procedure).

\item \label{action131} Removing unnecessary paths from $Q_v$ (Is-Relevant procedure).

\item  \label{action14} Updating the priority of $queue$'s paths ending at node $v$ (Push operation).

\item \label{action15} Inserting $p$ to the $queue$ (Push operation).

\end{enumerate}
Consider operation \ref{action11}. The Cost calculation of $p$ with respect to a single coefficient vector $\lambda \in \Lambda$ is done in $O(w)$. Since $|\Lambda|=K$, the action can be done in $O(K \cdot W)$.
\\
Consider operation \ref{action12}. The paths with minimal cost in $Q_v$ for each iteration can be calculated in an incremental manner, hence during the relevance check it is simply read in $O(1)$. 
\\
Consider operation \ref{action13}. It involves iterating through each of $p.Costs$, which can be done in $O(K)$.
\\
Consider operation \ref{action131}, It involves removing paths from $Q_v$ that are no longer part of $best_v$ (see definition \ref{best}). According to Lemma \ref{l4_1} $|Q_v|<K$, this action can be done in $O(K)$.
\\
Consider operation \ref{action14}. It involves checking if $p$'s cost is lower than that of each of the $queue$ paths ending at $v$. Note that, by Lemma \ref{l4_1}, the number of paths in $Q_v$ is at most $K$, i.e, this operation can be done in $O(K)$.
\\
Consider operation \ref{action15}. $p$ is added to the $queue$ data structure, which is done in $O(1)$. In case $p$ is inserted to the $queue$'s Priority Heap (Definition \ref{priorityHeap}), it has to be inserted to a Fibonacci Heap, which is also executed in $O(1)$ \cite{fibQueueDijkstraComplexity}.
\\
By Assumption 1,
$$W = O(1)$$
We thus get that a scan of a node that is not an $L$-node is done in $O(K)$, as required. 
\end{proof}
We proceed to calculate the time complexity of the scans of nodes that are not $L$-nodes.

\begin{lemma}\label{L6}
The time complexity of scans of nodes that are not $L$-nodes is given by $$O(|V| \cdot \log^2(|V|))$$
\end{lemma}
\begin{proof}
Recall that, according to the definition of $V_L$ (Section \ref{complexityDef}), the set of $L$-nodes is given by $V_L$, hence the set of nodes that are not $L$-nodes is given by $V \setminus V_L$.
\\

Since each node has at most $MaxDeg_-$ incoming edges and since during a single iteration each node is developed at most once (Lemma \ref{l1}), each node can be scanned at most $MaxDeg_-$ times during a single iteration, thus implying that the number of scans of nodes that are not $L$-nodes is upper-bounded by $\#Not L Nodes Scans$, where:
$$
\#Not L Nodes Scans =K \cdot |V \setminus V_L| \cdot MaxDeg_-$$
$$=  K \cdot (1-\alpha) \cdot V \cdot MaxDeg_-
$$
In Lemma \ref{l4} we have established that a single IDAQ scan is performed in $O(K)$. Therefore, the time complexity of scans of nodes that are not $L$-nodes is given by:
$$
O(\#Not L Nodes Scans \cdot K) = O(K^2 \cdot (1-\alpha) \cdot |V| \cdot MaxDeg_-).
$$
By Lemma \ref{a1}, the number of nodes that are not $L$-nodes is given by:
$$
O((1-\alpha) \cdot |V|)=O(log(|V|).
$$
By Assumption \ref{a2},  $$MaxDeg_-=O(\log(|V|).$$
By Assumption \ref{a3}:
$$
K=O(\sqrt{|V|}).
$$
Thus, the time complexity of scans of paths to nodes that are not $L$-nodes is given by:
$$O(|V| \cdot \log^2(|V|))$$
as required.
\end{proof}
Next, we employ the upper bound on the number of scans (as established in Theorem \ref{P6}), the time complexity of scanning a single $L$-node (as established in Lemma \ref{P5}) and the time complexity of scanning a single node that is not an $L$-node (as established in Lemma \ref{L6}) in order to calculate the total time complexity of IDAQ scans. We begin by proving an auxiliary lemma.
\\

First, we provide some intuition for the lemma.\\
Consider a $Path_{s,v}$ that is scanned during the execution of IDAQ. Recall that, in case $|Q_v| > \frac{K}{W}$, the Is-Relevant procedure uses the set $best_v$ (see Definition \ref{best}).\\
The $best_v$ set is calculated in an incremental manner (line 20 of the Is-Relevant procedure), however, for the first time it is being used, it has to be initialized.\\
In the following lemma, we shall calculate the time complexity of such an initialization during the execution of the IDAQ algorithm. \\
\begin{lemma}
\label{P8}
The time complexity of checking whether each path in $Path_{s,v}$  Is-Relevant for the first time where $Q_v >\frac{K}{W}$ is given by: $$O(|V| \cdot \log(|V|))$$
\end{lemma}
\begin{proof}
Consider some node that is not an $L$-node, say $v \notin V_L $. Consider the first time that the Is-Relevant procedure identifies more than $\frac{K}{W}$ paths in $Q_v$.
\\
The Is-Relevant procedure has to find all paths in $Q_v$ with minimal cost for each $\lambda \in \Lambda$ ($Best$ paths list). This involves calculating the cost of each path $p \in Q_v$ for each coefficient vector in $\lambda \in \Lambda$. This is done in:
$$O(|Q_v| \cdot K \cdot W)=O(\frac{K}{W} \cdot K \cdot W)=O(K^2)$$
and, due to Assumption \ref{a3}:
$$
=O(|V|)
$$

According to Lemma \ref{a1}, the number of nodes that are not $L$-nodes is given by: $$O(|V| \cdot (1-\alpha))=O(log(|V|).$$
This implies that identifying all nodes that are not $L$-nodes is done in:
$$O(|V| \cdot \log(|V|)).$$
\end{proof}
\begin{theorem}
\label{P9}
The total time complexity of scans in IDAQ is given by $$O(|E| \cdot L^2 +|V| \cdot \log^2(|V|))$$
\end{theorem}
\begin{proof}
In Theorem \ref{P6} we have established that the number of scans during an execution of the IDAQ algorithm is upper-bounded by  $|P_{Scanned}|$, where:
 \begin{equation}
     |P_{Scanned}| =  O(|E| \cdot L)
 \end{equation}
 In Lemma \ref{P5} we have established that a scan of an $L$-node is done in $O(L \cdot W)$. The number of scans of $L$-Nodes is upper-bounded by the total number of scans $|P_{Scanned}|$, therefore the scans of $L$-Nodes in IDAQ is done in:
 $$
 O(|P_{Scanned}| \cdot L \cdot W+|V| \cdot \log^2(|V|))
 $$
 $$
 =O(|E| \cdot L^2 \cdot W+|V| \cdot \log^2(|V|))
 $$
In Lemma \ref{L6} we have established that the scans nodes that are not $L$-nodes is done in $O(|V| \cdot \log^2(|V|))$.\\
In Lemma \ref{P8} we have established that the time complexity of checking whether each path in $Path_{s,v}$  Is-Relevant for the first time where $Q_v >\frac{K}{W}$ is given by: $$O(|V| \cdot \log(|V|))$$
Hence, the time complexity of scans in IDAQ is:
$$
 O(|E| \cdot L^2 \cdot W+|V| \cdot \log^2(|V|)) + 
 $$
 $$
 O(|V| \cdot \log^2(|V|)) + O(|V| \cdot \log(|V|))
$$
$$ = O(|E| \cdot L^2 \cdot W+|V| \cdot \log^2(|V|)+|V| \cdot \log(|V|))= $$
By Assumption \ref{a0}, $$W = O(1)$$ Thus we get,
$$
O(|E| \cdot L^2 + |V| \cdot \log^2(|V|))
$$

\end{proof}
\subsubsection{Complexity Analysis Part \ref{Part3} - Adapting the Queue}\label{part3Analysis}
We proceed to establish the time complexity of adapting the $queue$ between IDAQ iterations (part \ref{Part3}). We shall show that, even though this part is relatively time-consuming, it is still more efficient than developing paths in the Standard Algorithm, which is done in:
$$O(K \cdot |V| \cdot \log(|V|)$$

\begin{theorem}\label{AdaptThoreme}
Adapting the $queue$ between IDAQ iterations is done in: 
$$
O(K \cdot |V| \cdot L)=O(K \cdot |V| \cdot F(|V|))
$$
\end{theorem}
\begin{proof}
Consider the beginning of a new iteration, $i$. In order to adapt the $queue$, the following operations need to be executed:
\begin{enumerate}
    \item \label{action101}Setting each path's priority for iteration $i$.
    \item \label{action 102}Initialize a new priority $queue$ with the priorities calculated in \ref{action101}
\end{enumerate}

Consider operation \ref{action101}. According to definition \ref{stateCost}, this involves finding for each node $v \in V$ a path $p \in Path_{s,v}$ with minimal $Cost(p,\lambda i)$ among all the path in the $Q_{v}$. For a node that is not an $L$-node, this can be done in $O(1)$, since for such a node the $queue$ holds such a path for each iteration ($Best$ path list calculated by the Is-Relevant procedure). For an $L$-node, this can be done in $O(L \cdot W)$. In the worst case, i.e, when all nodes are $L$-nodes, this is done in $O(|V| \cdot L \cdot W)$.
\\

Consider operation \ref{action 102}. Using Fibonacci heaps for the priority queue, this is done in $O(V)$.
\\

In summary, we get that adapting the queue between iterations is done in:
$$
O(|V| \cdot L \cdot W)
$$
Since there is a total of $K$ iterations we get that the total time complexity of adapting the queue is:
$$
O(K \cdot |V| \cdot L \cdot W)=O(K \cdot |V| \cdot F(|V|)  \cdot W)
$$
By Assumption \ref{a0}, $$W = O(1)$$ thus we get
$$
=O(K \cdot |V| \cdot L )=O(K \cdot |V| \cdot F(|V|))
$$
\end{proof}

\subsubsection{Complexity Analysis Part \ref{Part4} - Build-Sets procedure}\label{part4Analysis}
At the end of the IDAQ algorithm, the Build-Sets procedure converts the list of discovered optimal paths with respect to any coefficient vector to a set of paths for each coefficient vector, as required from a solution of the problem.
\\
In this section, we calculate the time complexity of the  Build-Sets procedure (line 11 of algorithm \ref{idaqAlgo}).

\begin{theorem}\label{buildSettheorem}
The complexity of the Build-Sets procedure (line 11 in algorithm \ref{idaqAlgo}) is given by: 
$$
O(K \cdot |V| \cdot L)=O(\sqrt{|V|} \cdot |V| \cdot F(|V|))
$$
\end{theorem}
\begin{proof}
For each coefficient vector $\lambda_i \in \Lambda$, the Build-Sets procedure finds optimal paths ending in each node with respect to $\lambda_i$ by searching among the paths in the $optimal\_paths$ list.
\\

Recall that in Lemma \ref{optStatetheorem} we have established that each developed path is optimal, hence, the $optimal\_paths$ list, which is the input paths list for the Build-Sets procedure, holds only optimal paths.
\\

We begin by calculating the complexity of finding an optimal path with respect to $\lambda_i$ ending in a some node which is not an $L$-node, say $v \in V \setminus V_L$.
In Lemma \ref{l1} we have established that, at each iteration, IDAQ develops each node at most once. Recall that the number of iterations is given by $K$, i.e., the number of developed paths ending at any node is upper-bounded by $K$. Therefore, finding an optimal path with respect to $\lambda_i$ is done in $O(K \cdot W)$.
\\

We proceed to calculate the complexity of finding an optimal path with respect to $\lambda_i$ ending at some  $L$-node, say $v \in V_L$.  Due to Lemma \ref{lNodOpt}, for any $L$-node, say $v \in V_L$, there are at most $L$ different $Path_{s,v}$ that are optimal with respect to any coefficient vector in $\Lambda$. Therefore, finding the optimal path with respect to a single coefficient vector is done in $O(L \cdot W)$.
\\

We now calculate the complexity of finding an optimal path with respect to $\lambda_i$ ending at any node. Recall that the number of nodes that are not $L$-Nodes is bounded by $O(log(|V|))$ (Lemma \ref{a1}) and the number of $L$-Nodes is bounded by the total number of nodes, namely $|V|$. i.e, the time complexity is given by:
$$
O(log(|V|) \cdot K + |V| \cdot L )
$$

Finally, we are ready to calculate the complexity of the Build-Sets procedure. The Build-Sets procedure 
finds an optimal path with respect to any coefficient vector ending in any node. i.e, its time complexity is given by:
$$
O(K \cdot(log(|V|) \cdot K \cdot W  + |V| \cdot L ))=
$$
$$
O(log(|V|) \cdot K^2  \cdot W+ |V| \cdot L \cdot K \cdot W)
$$
By Assumption \ref{a0}, $$W = O(1)$$ Thus we get,
$$
=O(log(|V|) \cdot |V| + |V| \cdot F(|V|) \cdot \sqrt(|V|))
$$
$$
=O(|V| \cdot F(|V|) \cdot \sqrt(|V|))=O( K \cdot |V| \cdot L)
$$
as required.
\end{proof}
\subsubsection{IDAQ Time Complexity}
In Sections \ref{part1Analysis}, \ref{part2Analysis}, \ref{part3Analysis}, \ref{part4Analysis} we analyzed the time complexity of each of the four parts of the IDAQ algorithm, namely: developing paths, scanning path, adapting the $queue$ between iterations, and executing the Build-Sets procedure. This now allow us to establish the total time complexity of IDAQ, as presented in the following theorem.\\
\begin{theorem}\label{IDAQcomplexity}
The time complexity of IDAQ is given by:
$$
O(|E| \cdot F^{2}(|V|) +\sqrt{|V|} \cdot |V| \cdot F(|V|))=
$$$$
O(|E| \cdot L^2 +K \cdot |V| \cdot L)$$
\end{theorem}
\begin{proof}
By Theorem \ref{developTime},
the time complexity of developing paths is:
$$
O(L \cdot |V| \cdot \log(|V|)) =
O(F(|V|) \cdot |V| \cdot \log(|V|)).
$$
By Theorem \ref{P9}, the time complexity of scanning paths is: 
$$
O(|E| \cdot L^2 +|V| \cdot \log^2(|V|)).
$$
By Theorem \ref{AdaptThoreme}, the time complexity of adapting the $queue$ between IDAQ iterations is:
$$
O(K \cdot |V| \cdot L)=O(\sqrt{|V|} \cdot |V| \cdot F(|V|) ).
$$
By Theorem \ref{buildSettheorem}, the time complexity of executing the Build-Sets procedure is:
$$
O(K \cdot |V| \cdot L)=O(\sqrt{|V|} \cdot |V| \cdot F(|V|) ).
$$
We thus conclude that the time complexity of IDAQ is given by:
$$
O(F(|V|) \cdot |V| \cdot \log(|V|))+
$$$$
O(|E| \cdot L^2 +|V| \cdot \log^2(|V|))+
$$$$
O(\sqrt{|V|} \cdot |V| \cdot F(|V|))+
$$$$
O(\sqrt{|V|} \cdot |V| \cdot F(|V|))
$$$$
=O(|E| \cdot F^{2}(|V|) +\sqrt{|V|} \cdot |V| \cdot F(|V|))
$$$$
=O(|E| \cdot L^2 +K \cdot |V| \cdot L)
$$
\end{proof}
we are ready to formally claim that under the assumptions presented in section \ref{IDAQassumptionsSection}, IDAQ's time complexity (Theorem \ref{IDAQcomplexity}) is asymptotically lower than that of the Standard Algorithm.

\begin{theorem}\label{idaqBetterThorem}
The time complexity of IDAQ is asymptotically lower than that of the Standard Algorithm by a factor of: $$\Omega(\frac{\log{|V|}}{F(|V|)}$$
\end{theorem}
\begin{proof}
In Theorem \ref{IDAQcomplexity} we established that the time complexity of IDAQ is given by:
\begin{equation}\label{IdaqComExp}
O(\underbrace{L^2 \cdot |E|}_\text{(I1)} + \underbrace{ K \cdot |V| \cdot L|}_\text{(I2)}) = O(I_1 + I_2)    
\end{equation}
where:
$$
I_1 = L^2 \cdot |E|
$$
$$
I_2 =  K \cdot |V| \cdot L
$$
In Lemma \ref{l0} we established that the time complexity of the Standard Algorithm is given by:
$$
O(K \cdot W \cdot |E|+K \cdot |V| \cdot log|V|)
$$
By Assumption \ref{a0}, 
$$W = O(1)$$
Thus we get that the time complexity of the Standard Algorithm is given by:
\begin{equation}\label{SAComExp}
O(\underbrace{K \cdot |E|}_\text{(S1)} + \underbrace{K \cdot |V| \cdot log|V|}_\text{(S2)}) = O(S_1+S_2)
\end{equation}
where:
$$
S_1 =K \cdot |E|
$$
$$
S_2 = K \cdot |V| \cdot log|V|
$$
Recall that, according to Assumption \ref{A6},
$$
   L=O(F(|V|) )
$$
Also recall that according to Assumption \ref{a3}:
$$
K = O(\sqrt(|V|))
$$
Denote:
$$ \omega_1=   \frac{\sqrt{|V|}}{F^{2}(|V|)}$$
$$ \omega_2=  \frac{\log{|V|}}{F(|V|)}$$
Note that
\begin{equation}\label{I2Eq}
\begin{split}
O(I_2) \cdot \Theta(\omega_2)  & =O(K \cdot |V| \cdot L) \cdot \Theta (\frac{\log{|V|}}{F(|V|)})  \\
 & =O(K \cdot |V| \cdot F(|V|) \cdot \frac{\log{|V|}}{F(|V|)}) \\
 & = O(K \cdot |V| \cdot log|V|) \\
 & = O(S_2)
\end{split}
\end{equation}
In addition,
\begin{equation}\label{I1Eq}
\begin{split}
O(I_1) \cdot \Theta(\omega_1 ) & = O(L^2 \cdot |E| ) \cdot \Theta(\frac{\sqrt{|V|}}{F^{2}(|V|)} ) \\
& = O(L^2 \cdot |E|  \cdot \frac{\sqrt{|V|}}{F^{2}(|V|)} ) \\
& = O(F^{2}(|V|) \cdot |E|  \cdot \frac{\sqrt{|V|}}{F^{2}(|V|)} )\\
 & = O(\sqrt{|V|} \cdot |E|) = O(K \cdot |E|) \\
 & = O(S_1)
\end{split}
\end{equation}
Since $F$ is a sub-logarithmic function (Assumption \ref{A6}):
$$
\omega_1  = \frac{\sqrt{|V|}}{F^{2}(|V|)} > \Theta(\log(|V|)) 
$$
and
$$
\omega_2 < \Theta(\log(|V|))
$$
The following inequality holds:
\begin{equation}\label{omegaEq}
 \Theta(\omega_2) < \Theta(\omega_1)
\end{equation}
From expressions \ref{I1Eq} and \ref{omegaEq} we can conclude that the following holds:
\begin{equation}\label{I1Omega2}
\begin{split}
O(I_1 ) \cdot \Theta( \omega_2)  &  < O(I_1) \cdot \Theta(\omega_1) =  O(S_1)
\end{split}
\end{equation}
Finally, from expressions \ref{I2Eq} and \ref{I1Omega2} we can conclude that
\begin{equation}
\begin{split}
O(I_1 + I_2) \cdot \Theta(\omega_2)  & =O( I_1) \cdot \Theta( \omega_2 ) + O( I_2) \cdot \Theta( \omega_2)    \\
 & = O(S_1) + O( S_2) \\
 & = O(S_1+S_2)
\end{split}
\end{equation}
Since the complexity of IDAQ is given by $O(I_1+I_2)$ (expression \ref{IdaqComExp}) and the complexity of the Standard Algorithm is given by $O(S_1+S_2)$ (expression \ref{SAComExp}), we conclude that IDAQ's complexity is lower than that of the Standard Algorithm by at least a factor of $\omega_2$.
Thus, we get that, indeed, the time complexity of IDAQ is asymptotically lower than that of the Standard Algorithm by a factor of: $$\Omega(\omega_2)  = \Omega(\frac{\log(|V|)}{F(|V|)})$$
\end{proof}
In the previous Theorem \ref{idaqBetterThorem}, we established that the time complexity of IDAQ is asymptotically lower than that of the Standard Algorithm by a factor of $\Omega(\frac{\log{|V|}}{F(|V|)}$.
Recall that, according to Definition \ref{L}, $L = O(F(|V|))$ corresponds to the maximal number of Pareto paths for most nodes, which can be expected to be relatively small in practice \cite{paretoIsFeasible}. 
For example, if $F(|V|) = O(1)$, the factor of improvement is given by $\Omega(\log{|V|})$; 
if $F(|V|) = \log{\log{|V|}}$, 
it is $\Omega( \frac{\log{|V|}}{\log{\log{|V|}}})$;
and if $F(|V|) = \log^{1-\beta}{|V|}$, for some $0<\beta<1$, then it is $\Omega( \log^{\beta}(|V|)$. 

\section{Simulation Study}
In this section, we present computational experiments conducted in order to assess the performance of IDAQ (Algorithm \ref{idaqAlgo}) in comparison to the Standard Algorithm (Algorithm \ref{naiveAlgoPseudo}). Both algorithms have been implemented in a MATLAB environment. All experiments were conducted on a PC with 32GB RAM and 5$th$ Generation Intel® Core™ i5 Processor. In Section \ref{randomExp}, we describe experiments conducted using randomly generated $MOG$s. In Section \ref{practicleExp} ,we describe experiments conducted on $MOG$s generated using actual data representing a more practical setting (provided by Open Street  \cite{openstreetmap}). 
As shall be presented, in each experiment, we obtained identical results for IDAQ and for the Standard Algorithm, in terms of the quality (optimality) of the solution, as indeed implied by the established correctness of IDAQ (Theorem \ref{idaqcorrect}). In each experiment, we measured the performance of each algorithm in terms of running time.

\subsection{Random Generated Experiments}\label{randomExp}
In this section, we demonstrate the advantage of using IDAQ to solve a randomly generated instance of the $MOWSP$ problem (Problem \ref{mowsp}). We describe two different experiments (Experiment 1 and Experiment 2) conducted in order to asses IDAQ's improvement in performance, in terms of running time, with respect to the Standard Algorithm. In Section \ref{mogsgen} we describe the $MOG$ (Definition \ref{MOG}) instances employed by the experiments (we used the same $MOG$s in the two experiments). In Sections \ref{Exp1} and \ref{Exp2} we describe the conducted experiments. 

\subsubsection{$MOG$ instances for Experiments 1 and 2}\label{mogsgen}
In this section we describe the generation of $MOG$s employed by Experiment 1 and Experiment 2.
\\First, we generated random Waxman graphs \cite{waxman} using various parameters. The properties of the generated graph and generation parameters are described in Table \ref{tab:table1}. In order to generate an $MOG$ out of our generated Waxman graphs, we selected for each edge, uniformly at random, five objectives (see Definition \ref{MOG}), each assuming a value between $0$ and $1$.

\begin{table}[t]
\begin{center}
 \begin{tabular}{||c c c||} 
 \hline
Number of Nodes & Number Of Edges & Density ($\frac{|E|}{|V|^2)}$) \\ [0.5ex] 
 \hline\hline
 221 & 8344 & 0.1708 \\ 
 \hline
 287 & 14290 & 0.1734  \\
 \hline
 236 & 9682 & 0.1738  \\
 \hline
 245 & 10474 & 0.1744  \\
 \hline
 226 & 8566 & 0.1677  \\ 
 \hline
  266 & 12306 & 0.17394 \\ 
 \hline
  238 & 9716 & 0.1715 \\ 
 \hline
  233 & 9370 & 0.1725 \\ [1ex] 
 \hline
\end{tabular}
\caption { Randomly Generated Waxman graph parameters, used in the experiments conducted in Section \ref{randomExp}. The graphs were generated using the following Waxman model parameters: $Domain = \{[0,1],[0,0.1]\}$; $\lambda = 5000$; $\alpha = 4$; $\beta = 0.03$;}
\label{tab:table1}
\end{center}
\end{table}

\subsubsection{Experiment 1}\label{Exp1}
In this section, we describe the experiment conducted in order to asses IDAQ's improvement in running time with respect to the Standard Algorithm. 

Fist, we generate a set of $MOG$s, as described in Section \ref{mogsgen}).\\
Next, we construct  a $MOWSP$ out of each $MOG$ by selecting uniformly at random $K$ coefficient vectors (see definition \ref{lambdaEq} (recall that $K=|\Lambda$|)). Each coefficient is a randomly uniformly generated number between 0.1 and 1.1. we compared between the running times of IDAQ and the Standard Algorithm under various values of $K$. For each value of $K$ we calculated the average running time of IDAQ and the Standard Algorithm (Figure \ref{f1}).  As can be seen, IDAQ exhibits better performance (in terms of running time), and its advantage is significant (up to an improvement of about $50\%$) for large numbers of coefficient vectors.
As described in Section \ref{IDAQDesc}, similarly to the Standard Algorithm, IDAQ is an iterative algorithm, which produces at each iteration at most a single solution for each $v \in V$. However, unlike the Standard Algorithm, IDAQ shares knowledge between its iterations. Therefore, we expect that in $MOWSP$s with a large size of the coefficient vector, namely with a large $K$, IDAQ would exhibit better performance, as indeed demonstrated in Figure \ref{f1}.\\

\begin{figure}[ht!]
\includegraphics[scale=1.0]{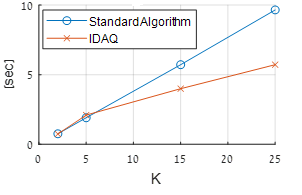}
\caption{Average running times of the Standard Algorithm and IDAQ  for the MOWSP problem with different values of $K$.}
\label{f1}
\end{figure}

\subsubsection{Experiment 2}\label{Exp2}
In this section we describe an experiment conducted in order to compare between the running times of IDAQ and the Standard Algorithm in scenarios in which IDAQ develops a significantly smaller number of paths than the Standard Algorithm. In other words, we investigated IDAQ's performance improvement in scenarios where there is considerable knowledge that can be shared between the algorithm's iterations (in particular, considerably more than in the Waxman topologies of Experiment 1).\\

We found out that an easy way to control the number of IDAQ developed paths is through the $MOWSP$'s coefficient vectors. Intuitively, using coefficient vectors that exhibit similarity increases the probability that a path that is optimal with respect to a specific coefficient vector would be optimal also with respect to other coefficient vectors, thus lowering the amount of paths developed by IDAQ. 
\\

The experiment runs as follows.
We construct a set of $MOWSP$ in a similar manner to Experiment 1 (Section \ref{Exp1}). We use the same set of $MOG$s, but now create the coefficient vectors with a small change: instead of setting each coefficient to a uniformly distributed number between $0.1$ to $1.1$ (as in Experiment 1), we now set it to a uniformly distributed number between $0.5$ to $1.1$ (i.e., larger values than in Experiment 1), which results in more similar coefficient vectors. Figure \ref{f1_1} demonstrates that, in such scenarios IDAQ exhibits much better performance than the Standard Algorithm. 

\begin{figure}[ht!]
\includegraphics[scale=1.0]{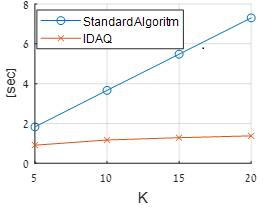}
\caption{Standard Algorithm and IDAQ average running time on $MOWSP$ problem with a different number of coefficient vectors $K$. The coefficient vectors were selected so that iterations are more correlated than in Experiment 1.}
\label{f1_1}
\end{figure}

\subsection{Practical Application Experiments}\label{practicleExp}
In this section we describe several experiments conducted on $MOWSP$ generated from actual data.\\
We consider an application that finds several "optimal" routes for a bicycle rider in Manhattan (New York, USA). To generate our $MOG$, we used data provided by Open Street Map \cite{openstreetmap}. Each edge's (=road) objective is determined by the following:
\begin{description}
    \item \textbf{C1} - Road distance.
    \item \textbf{C2} -  If bicycle road : $\frac{C1 }{ 2}$ , $C1$ otherwise.
    \item \textbf{C3} -  If road not close to highway : $\frac{C1 }{ 2}$ , $C1$ otherwise.
    \item \textbf{C3} -  If road not close to buildings : $\frac{C1 }{ 2}$ , $C1$ otherwise.
\end{description}
In Table \ref{table2} we present seven considered coefficient vectors. 
\begin{table}[]
 \begin{center}
 \begin{tabular}{||c c c c||} 
 \hline
 C1  & C2  & C2 & C3 \\ [0.5ex] 
 \hline\hline
 13.06  & 0.17	& 0.13	& 0.21 \\ 
 \hline
 0.13	& 16.98 &	0.13 &	0.21 \\
 \hline
 0.13 & 0.17 &	13.17	& 0.21 \\
 \hline
 0.13 & 0.17 &	0.13 &	21.11  \\
 \hline
 4.28 & 7.45	& 1.4 &	3.25 \\ 
 \hline
  3.89 & 6.12 &	1.5 &	5.45 \\ 
 \hline
 6.23 & 8.27 & 0.61 &	0.45 \\ [1ex]
 \hline
\end{tabular}
\caption{Example of weigh vectors in Experiment 2, some were selected manually (lines 1-4), other at random (lines 5-7).}
\label{table2}
\end{center}
\end{table}
In Figure \ref{f2} we depict the routes identified by IDAQ on a specific problem ($25$ coefficient vectors).   
\begin{figure}[ht!]
\includegraphics[scale=0.6]{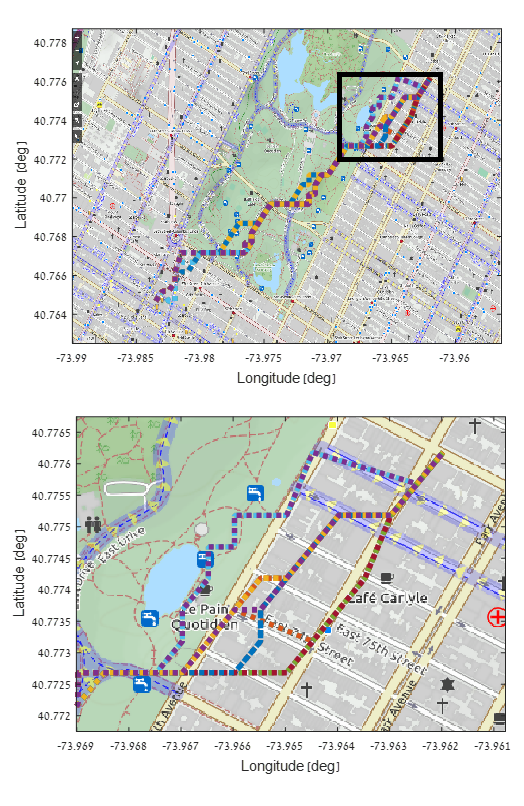}
\caption{Optimal routes selected by IDAQ on a generated MOWSP problem. Lower part of the figure zooms on areas where routes set apart.}
\label{f2}
\end{figure}

Figure \ref{f4} demonstrates that, similarly to the randomly generated experiments (Section \ref{randomExp}), in the current experiments too IDAQ exhibits better performance than the Standard Algorithm.

\begin{figure}[ht!]
\includegraphics[scale=0.9]{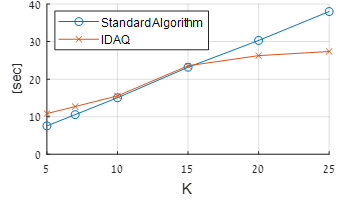}
\caption{Standard Algorithm and IDAQ average running time on MOWSP problem generated using actual data (by Open Street Map) with different amount of coefficient vectors $K$.}
\label{f4}
\end{figure}

\newpage

\section{Conclusion}
We investigated the fundamental problem of routing with multiple objectives. More specifically, we considered the problem of providing several routes that minimize different optimization criteria. While this can be simply achieved by consecutively executing a standard shortest path algorithm, in case of a large number of different optimization criteria this may require an excessively large number of executions, thus incurring a prohibitively large running time.
\\

Our major contribution is a novel efficient algorithm for the considered problem, namely the IDAQ algorithm. Similarly to the standard-approach algorithm, IDAQ iteratively searches for routes, one per optimization criteria; however, instead of executing each iteration independently, it reduces the average running time by skilfully sharing information among the iterations. By doing so, it exploits the similarity among optimal routes with respect to different optimization criteria, so as to improve the performance of the solution scheme. 
\\

We showed that both IDAQ and the standard algorithm provide an optimal solution for the considered problem. We then showed that, under reasonable assumptions, IDAQ typically provides considerably lower computational complexity than that of the standard algorithm. We confirmed this finding through several computational experiments on both randomly generated settings, as well as settings that correspond to real-world environments (specifically using data generated from Open Street Map \cite{openstreetmap}). 
\\

Several important issues are left for future work. One is to investigate whether assumptions employed in this paper can be applied to speed up other multi-objective routing solution schemes.
For instance, in this study, we presented an efficient algorithm that takes advantage of small Pareto sets to gain significant speed-up in comparison to a standard approach algorithm that consecutively executes a (standard) shortest path algorithm. We employed the assumption that, even though the number of Pareto paths can be non-polynomial, it is often relatively small in practice \cite{paretoIsFeasible}. It is of interest to investigate whether such an assumption can be applied to speed up other multi-objective routing schemes, such as Resource Constrained Shortest Path algorithms or Fully Polynomial Time Approximation Schemes (FPTAS). 
\\

Another interesting direction is to investigate whether similar results can be obtained in a setting where more than one computational unit is allocated to solve the MOWSP problem. A major advantage of the standard approach algorithm over our proposed algorithm is that it is easy to break it into several processes: each iteration (in other words, execution of the Dijkstra algorithm) can be executed independently, hence different computational units can take care of different iterations, simultaneously.
It is of interest to investigate whether in the case where more than a single computational unit is available, an efficient algorithm can be established to solve the MOWSP problem while maintaining a significant speedup in comparison to a standard approach algorithm. This is, of course, quite challenging since our approach is based on iterative steps where each iteration depends on the previous one. 
It might be interesting to consider an iterative algorithm where each iteration executes a distributed version of a shortest-path algorithm, e.g. Bellman-Ford's \cite{BF}. This may allow us to use similar ideas to those implemented in IDAQ while taking advantage of multiple available computation units.
\\

Another important aspect is to compare our approach with sub-optimal algorithms that provide an estimation of the Pareto set. In particular, it is of interest to investigate under which conditions our approach provides better results (in terms of the quality of the solution or running time). The latter has to consider that, unlike Pareto sub-optimal algorithms, our approach depends on pre-selected coefficient vectors.
\\

Last, in some settings, a heuristic scheme can be applied to speed up a shortest path search, while still providing optimal solutions (e.g, the A* algorithm \cite{aStar} or other informed search algorithms \footnote{I.e., an algorithm guided by some heuristic}). In the case where an admissible heuristic is supplied for each coefficient vector, a standard informed algorithm can be proposed that is based on executions of an informed search algorithm (instead of Dijkstra's). It is of interest to investigate whether an efficient informed algorithm can be proposed, based on similar ideas to those implemented in IDAQ, while reaching a similar speedup factor in comparison to the standard informed algorithm.

\bibliographystyle{plain}
\bibliography{\jobname.bib}
\newpage

\end{document}